\documentclass[draftcls,onecolumn,11pt]{IEEEtran}

\usepackage{units}
\usepackage{amsmath}
\usepackage{graphicx}
\usepackage{amssymb}
\usepackage{stfloats}
\makeatletter

\providecommand{\tabularnewline}{\\}

\newtheorem{example}{Example}
\newtheorem{condition}{Condition}
\newtheorem{definitn}{Definition}
\newtheorem{prop}{Proposition}
\newtheorem{remrk}{Remark}
\newtheorem{lemma}{Lemma}
\newtheorem{cor}{Corollary}
\newtheorem{thm}{Theorem}

\newtheorem{alg}{Algorithm}

\usepackage{bbm}

\makeatother

\begin{document}

\title{On Multiple Decoding Attempts for Reed-Solomon Codes: A Rate-Distortion
Approach}

\author{
 Phong S. Nguyen, \IEEEmembership{Student Member, IEEE}, Henry D. Pfister, \IEEEmembership{Member, IEEE}, 
and \\Krishna R. Narayanan, \IEEEmembership{Senior Member, IEEE}
\thanks{
This work was supported by the National Science
Foundation under Grant No. 0802124. The work of P. S. Nguyen was also
supported in part by a Vietnam Education Foundation fellowship. The material in this paper was presented
in part at the 47th Annual Allerton Conference on Communications, Control and Computing, Monticello, IL, October 2009
and in part at IEEE International Symposium on Information Theory (ISIT), Austin, TX, June 2010.
}
\thanks{The authors are with the Department of Electrical and Computer
Engineering, Texas A\&M University, College Station, TX 77843, USA (email: psn@tamu.edu; hpfister@tamu.edu; krn@tamu.edu).}

}
\maketitle

\begin{abstract}
One popular approach to soft-decision decoding of Reed-Solomon (RS)
codes is based on using multiple trials of a simple RS
decoding algorithm in combination with erasing or flipping
a set of symbols or bits in each trial. This paper presents
a framework based on rate-distortion (RD) theory to analyze these multiple-decoding
algorithms. By defining an appropriate distortion measure
between an error pattern and an erasure pattern, the successful decoding
condition, for a single errors-and-erasures decoding trial, becomes
equivalent to distortion being less than a fixed threshold.
Finding the best set of erasure patterns also turns into a covering
problem which can be solved asymptotically by rate-distortion theory.
Thus, the proposed approach can be used to understand the asymptotic
performance-versus-complexity trade-off of multiple errors-and-erasures
decoding of RS codes. 

This initial result is also extended a few directions.
The rate-distortion exponent (RDE) is computed to give more precise results for moderate blocklengths.
Multiple trials of algebraic soft-decision (ASD) decoding are analyzed using this framework.
Analytical and numerical computations of the RD and RDE functions are also presented.
Finally, simulation results show that sets of erasure patterns designed using
the proposed methods outperform other algorithms with the same number of decoding trials. 
\end{abstract}


\section{Introduction}

\PARstart{R}{eed-Solomon} (RS) codes are among the most popular error-correcting
codes in communication and data storage systems. An $(N,K)$ RS code
of length $N$ and dimension $K$ is a maximum distance separable
(MDS) linear code with minimum distance $d_{\min}=N-K+1$. RS codes
have efficient hard-decision decoding (HDD) algorithms, such as the
Berlekamp-Massey (BM) algorithm, which can correct up to $\big\lfloor\frac{d_{\min}-1}{2}\big\rfloor$
errors.

Since the discovery of RS codes \cite{Reed-jsim60}, researchers have
spent a considerable effort on improving the decoding performance
at the expense of complexity. A breakthrough result of Guruswami and
Sudan (GS) introduced an algebraic hard-decision list-decoding algorithm,
based on bivariate interpolation and factorization, that
can correct errors well beyond half the minimum distance of the code
\cite{Guruswami-it99}. Nevertheless, HDD algorithms do not fully
exploit the information provided by the channel output. Koetter and
Vardy (KV) later extended the GS decoder to an algebraic soft-decision
(ASD) decoding algorithm by converting the probabilities observed
at the channel output into algebraic interpolation conditions in terms
of a multiplicity matrix \cite{Koetter-it03}. 

The GS and KV algorithms, however, have significant computational
complexity. Therefore, multiple runs of errors-and-erasures and errors-only
decoding with some low-complexity algorithm, such as the BM algorithm,
has renewed the interest of researchers. These algorithms
use the soft-information available at the channel output to
construct a set of either erasure patterns \cite{Forney-it66,Lee-globecom08},
test patterns \cite{Chase-it72}, or patterns combining both \cite{Tang-comlett01,Tokushige-ieice06}
and then attempt to decode using each pattern. Techniques have also
been introduced to lower the complexity per decoding trial in \cite{Bellorado-it10,Xia-com07,Xia-icc08}.
Other soft-decision decoding algorithms for RS codes include \cite{Jiang-it06,Fossorier-it95}
that use the binary expansion of RS codes to work on the bit-level.
In \cite{Jiang-it06}, belief propagation is run while the parity-check
matrix is iteratively adapted on the least reliable basis. Meanwhile,
\cite{Fossorier-it95} adapts the generator matrix on the most reliable
basis and uses reprocessing techniques based on ordered statistics.

In the scope of multiple errors-and-erasures decoding, there have
been several algorithms proposed that use different erasure codebooks
(i.e., different sets of erasure patterns). After running the errors-and-erasures
decoding algorithm multiple times, each time using one erasure pattern
in the set, these algorithms produce a list of candidate codewords,
whose size is usually small, and then pick the best codeword on this
list. The common idea of constructing the set of erasure patterns
in these multiple errors-and-erasures decoding algorithms is to erase
some of the least reliable symbols since those symbols are more prone
to be erroneous. The first algorithm of this type is called Generalized
Minimum Distance (GMD) \cite{Forney-it66} and it repeats errors-and-erasures
decoding while successively erasing an even number of the least reliable
positions (LRPs) (assuming that $d_{\min}$ is odd). More recent work
by Lee and Kumar \cite{Lee-globecom08} proposes a soft-information
successive (multiple) error-and-erasure decoding (SED) that achieves
better performance but also increases the number of decoding attempts.
Literally, the Lee-Kumar's SED$(l,f)$ algorithm runs multiple errors-and-erasures
decoding trials with every combination of an even number $\leq f$
of erasures within the $l$ LRPs. 

A natural question that arises is how to construct the {}``best''
set of erasure patterns for multiple errors-and-erasures decoding.
Inspired by this, we first develop a rate-distortion (RD) framework
to analyze the asymptotic trade-off between performance and complexity
of multiple errors-and-erasures decoding of RS codes. The main idea
is to choose an appropriate distortion measure so that the decoding
is successful if and only if the distortion between the error pattern
and erasure pattern is smaller than a fixed threshold. After that,
a set of erasure patterns is generated randomly (similar to a random
codebook generation) in order to minimize the expected minimum distortion. 

One of the drawbacks in the RD approach is that the mathematical framework
is only valid as the block-length goes to infinity. Therefore, we
also consider the natural extension to a rate-distortion exponent
(RDE) approach that studies the behavior of the probability, $p_{e}$,
that the transmitted codeword is not on the list as a function of
the block-length $N$. The overall error probability can be approximated
by $p_{e}$ because the probability that the transmitted codeword
is on the list but not chosen is very small compared to $p_{e}$.
Hence, our RDE approach essentially focuses on maximizing the exponent
at which the error probability decays as $N$ goes to infinity. The
RDE approach can also be considered as the generalization of the RD
approach since the latter is a special case of the former when the
rate-distortion exponent tends to zero. Using the RDE analysis, this
approach also helps answer the following two questions: (i) What is
the minimum error probability achievable for a given number of decoding
attempts (or a given size of the set of erasure patterns)? (ii) What
is the minimum number of decoding attempts required to achieve a certain
error probability? 

The RD and RDE approaches are also extended beyond conventional errors-and-erasures
decoding to analyze multiple-decoding for decoding schemes such as
ASD decoding. It is interesting to note that the RDE approach for
ASD decoding schemes contains the special case where the codebook
has exactly one entry (i.e., ASD decoding is run only once). In this case,
the distribution of the codebook that maximizes the exponent implicitly
generates the optimal multiplicity matrix. This is similar to the
line of work \cite{Parvaresh-isit03,El-Khamy-dimacs05,Ratnakar-it05,Das-isit09}
where various researchers solve for a multiplicity matrix that minimizes
the error probability obtained by either using a Gaussian approximation
\cite{Parvaresh-isit03}, applying a Chernoff bound \cite{El-Khamy-dimacs05,Ratnakar-it05},
or using Sanov's theorem \cite{Das-isit09}.

Finally, we propose a family of multiple-decoding algorithms based
on these two approaches that achieve better performance-versus-complexity
trade-off than other algorithms.

\subsection{Outline of the paper}

The paper is organized as follows. In Section~\ref{sec:MultipleBMA},
we design an appropriate distortion measure and present a rate-distortion
framework, for both the RD and RDE approaches, to analyze the performance-versus-complexity
trade-off of multiple errors-and-erasures decoding of RS codes. Also
in this section, we propose a general multiple-decoding algorithm
that can be applied to errors-and-erasures decoding. Then, in Section~\ref{sec:Computing-RD},
we discuss numerical computations of RD and RDE functions together
with their complexity analyses which are needed for the proposed algorithm.
In Section~\ref{sec:Multiple ASD}, we analyze both bit-level and
symbol-level ASD decoding and design distortion measures compatible
with the general algorithm. A closed-form analysis of some RD and
RDE functions is presented in Section~\ref{sec:Closed-Form-Analysis-of}.
Next, in Section~\ref{sec:Ext-and-Gen}, we offer some extensions
that combine covering codes with random codes and also consider the
case of a single decoding attempt. Simulation results are presented
in Section~\ref{sec:Simulation-results} and, finally, conclusions
are provided in Section~\ref{sec:Conclusion}.

\section{A RD Framework For Multiple Errors-and-Erasures Decoding\label{sec:MultipleBMA}}

In this section, we first set up a rate-distortion framework to analyze
multiple attempts of conventional hard decision errors-and-erasures
decoding.

Let $\mathbb{F}_{m}$ with $m=2^{\eta}$ be the Galois field with
$m$ elements denoted as $\alpha_{1},\alpha_{2},\ldots,\alpha_{m}$.
We consider an $(N,K)$ RS code of length $N$, dimension $K$ over
$\mathbb{F}_{m}$. Assume that we transmit a codeword $\mathbf{c}=(c_{1},c_{2},\ldots,c_{N})\in\mathbb{F}_{m}^{N}$
over some channel and receive a vector $\mathbf{r}=(r_{1},r_{2},\ldots,r_{N})\in\mathcal{Y}^{N}$
where $\mathcal{Y}$ is the received alphabet for a single RS symbol.
While our approach can be applied to much more general channels, our
simulations focus on the Additive White Gaussian Noise (AWGN) channel
and two common modulation formats, namely BPSK and $m$-QAM. Correspondingly,
we use $\mathcal{Y}=\mathbb{R}^{\eta}$ for BPSK and $\mathcal{Y}=\mathbb{R}^{2}$
for $m$-QAM. For each codeword index $i$, let $\varphi_{i}:\{1,2,\ldots,m\}\rightarrow\{1,2,\ldots,m\}$
be the permutation given by sorting $\pi_{i,j}=\Pr(c_{i}=\alpha_{j}|r_{i})$
in decreasing order so that $\pi_{i,\varphi_{i}(1)}\geq\pi_{i,\varphi_{i}(2)}\geq\ldots\geq\pi_{i,\varphi_{i}(m)}$.
Then, we can specify $y_{i,j}=\alpha_{\varphi_{i}(j)}$ as the $j$-th
most reliable symbol for $j=1,\ldots,m$ at codeword index $i$. To
obtain the reliability of the codeword positions (indices), we construct
the permutation $\sigma:\{1,2,\ldots,N\}\rightarrow\{1,2,\ldots,N\}$
given by sorting the probabilities $\pi_{i,\varphi_{i}(1)}$ of the
most likely symbols in increasing order.%
\footnote{Other measures such as entropy or the average number of guesses might
improve Algorithm B in Section \ref{sec:Proposed-Algorithm}.%
} Thus, codeword position $\sigma(i)$ is the $i$-th LRP. These above
notations will be used throughout this paper.
\begin{example}
Consider $N=3$ and $m=4$. Assume that we have the probability $\pi_{i,j}$
written in a matrix form as follows:\[
\mathbf{\Pi}=\left(\begin{array}{ccc}
0.01 & 0.01 & \mathbf{0.93}\\
\mathbf{0.94} & 0.03 & 0.04\\
0.03 & \mathbf{0.49} & 0.01\\
0.02 & 0.47 & 0.02\end{array}\right) \text{where}\,\,\pi_{i,j}=[\mathbf{\Pi}]_{j,i}.\]
then $\varphi_{1}(1,2,3,4)=(2,3,4,1)$, $\varphi_{2}(1,2,3,4)=(3,4,2,1)$,
$\varphi_{3}(1,2,3,4)=(1,2,4,3)$ and $\sigma(1,2,3)=(2,3,1)$.\end{example}
\begin{condition}
\label{con:BMAerr-n-era}(Classical decoding threshold, see \cite{Lin-1983,Blahut-2003}):
If $e$ symbols are erased, a conventional hard-decision errors-and-erasures
decoder such as the BM algorithm is able to correct $\nu$ errors
in unerased positions if and only if \begin{equation}
2\nu+e<N-K+1.\label{eq:scbma_0}\end{equation}

\end{condition}

\subsection{Conventional error patterns and erasure patterns.}
\begin{definitn}
\label{def:(Conv. patterns)}(Conventional error patterns and erasure
patterns) We define $x^{N}\in\mathbb{Z}_{2}^{N}\triangleq\{0,1\}^{N}$
and $\hat{x}^{N}\in\mathbb{Z}_{2}^{N}$ as an error pattern and an
erasure pattern respectively, where $x_{i}=0$ means that an error
occurs (i.e., the most likely symbol is incorrect) and $\hat{x}_{i}=0$
means that the symbol at index $i$ is erased (i.e., an erasure is
applied at index $i$). $X^{N}$ and $\hat{X}^{N}$ will be used to
denote the random vectors which generate the realizations $x^{N}$
and $\hat{x}^{N}$, respectively.\end{definitn}
\begin{example}
If $d_{\min}$ is odd then the GMD algorithm corresponds to the set\[
\{111111\ldots,001111\ldots,000011\ldots,\ldots,\underbrace{00\ldots0}_{d_{\min}-1}11\ldots1\}\]
of erasure patterns. Meanwhile, the SED$(3,2)$ uses the following
set \[
\{\underline{111}111\ldots,\underline{001}111\ldots,\underline{010}111\ldots,\underline{100}111\ldots\}.\]
Here, in each erasure pattern, the letters are written in increasing
reliability order of the codeword positions.
\end{example}
Let us revisit the question of how to construct the best set of erasure
patterns for multiple errors-and-erasures decoding. First, it can
be seen that a multiple errors-and-erasures decoding succeeds if the
condition (\ref{eq:scbma_0}) is satisfied during at least one round
of decoding. Thus, our approach is to design a distortion measure
that converts the condition (\ref{eq:scbma_0}) into a form where
the distortion between an error pattern $x^{N}$ and an erasure pattern
$\hat{x}^{N}$, denoted as $d(x^{N},\hat{x}^{N})$, is less than a
fixed threshold.
\begin{definitn}
Given a \emph{letter-by-letter }distortion measure $\delta$, the
distortion between an error pattern $x^{N}$ and an erasure pattern
$\hat{x}^{N}$ is defined by\[
d(x^{N},\hat{x}^{N})=\sum_{i=1}^{N}\delta(x_{i},\hat{x}_{i}).\]
\end{definitn}
\begin{prop}
\label{prop:bma1}If we choose the \emph{letter-by-letter} distortion
measure $\delta:\mathcal{X}\times\hat{\mathcal{X}}\rightarrow\mathbb{R}_{\geq0}$,
where in this case $\mathcal{X}=\hat{\mathcal{X}}=\mathbb{Z}_{2}$,
as follows:\begin{equation}
\begin{array}{cc}
\delta(0,0)=1, & \delta(0,1)=2,\\
\delta(1,0)=1, & \delta(1,1)=0,\end{array}\label{eq:dstfnBMA}\end{equation}
then the condition (\ref{eq:scbma_0}) for a successful errors-and-erasures
decoding is equivalent to\begin{equation}
d(x^{N},\hat{x}^{N})<N-K+1\label{eq:distorthreshold}\end{equation}
where the distortion is less than a fixed threshold.\end{prop}
\begin{proof}
First, we define \[\chi_{j,k}\triangleq\left|\{i\in\left\{ 1,2,\ldots,N\right\} :x_{i}=j,\hat{x}_{i}=k\}\right|\]
to count the number of $(x_{i},\hat{x}_{i})$ pairs equal to $(j,k)$
for every $j\in\mathcal{X}$ and $k\in\hat{\mathcal{X}}$. With the
chosen distortion measure, we have \[d(x^{N},\hat{x}^{N})=2\chi_{0,1}+\chi_{0,0}+\chi_{1,0}.\]
Noticing that $e=\chi_{0,0}+\chi_{1,0}$ and $\nu=\chi_{0,1}$, the
condition (\ref{eq:scbma_0}) for one errors-and-erasures decoding
attempt to succeed becomes $2\chi_{0,1}+\chi_{0,0}+\chi_{1,0}<N-K+1$
which is equivalent to $d(x^{N},\hat{x}^{N})<N-K+1$. 
\end{proof}
Next, we try to maximize the chance that this successful decoding
condition is satisfied by at least one of the decoding attempts (i.e.,
$d(x^{N},\hat{x}^{N})<N-K+1$ for at least one erasure pattern $\hat{x}^{N}$).
Mathematically, we want to build a set $\mathcal{B}$ of no more than
$2^{R}$ erasure patterns $\hat{x}^{N}$ that achieves the maximum
\begin{equation*}
\max_{\mathcal{B}:|\mathcal{B}|\leq2^{R}}\Pr\left\{ \min_{\hat{x}^{N}\in\mathcal{B}}d(X^{N},\hat{x}^{N})<N-K+1\right\} .\label{eq:ProbStatement}\end{equation*}
Solving this problem exactly is very difficult.
However, one can observe that it is a covering problem where tries to
cover the most-likely error patterns using a fixed number of spheres
centered at the chosen erasure patterns. This view leads to two asymptotic
solutions of the problem based on rate-distortion theory. Taking this
point of view, we view the error pattern $x^{N}$ as a source sequence
and the erasure pattern $\hat{x}^{N}$ as a reproduction sequence.

\begin{figure}[t]
\begin{centering}
\includegraphics[scale=0.85]{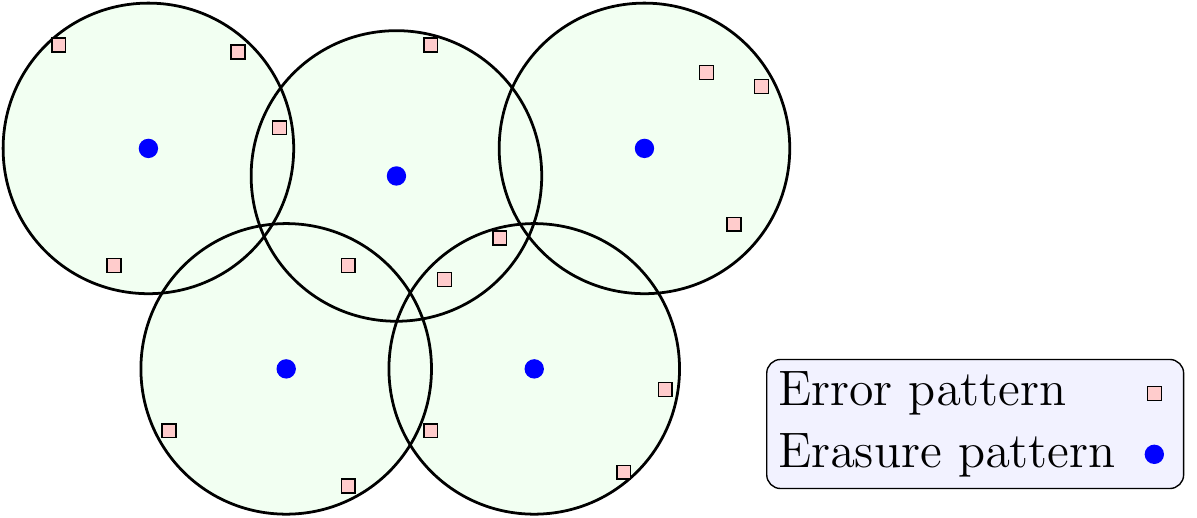}
\par\end{centering}

\centering{}\caption{Pictorial illustration of a covering problem}

\end{figure}

\subsubsection{RD approach}

Rate-distortion theory (see \cite[Chapter 13]{Cover-1991}) characterizes
the trade-off between $\bar{R}$ and $\bar{D}$ such that sets $\mathcal{B}$
of $2^{N\bar{R}}$ reproduction sequences exist (and can be generated
randomly) so that\[
\lim_{N\rightarrow\infty}\frac{1}{N}E_{X^{N},\mathcal{B}}\left[\min_{\hat{x}^{N}\in\mathcal{B}}d(X^{N},\hat{x}^{N})\right]<\bar{D}.\]
Under mild conditions, this implies that, for large enough $N$, we
have\[
\min_{\hat{x}^{N}\in\mathcal{B}}d(X^{N},\hat{x}^{N})<N\bar{D}\]
with high probability. Here, $\bar{R}$ and $\bar{D}$ are closely
related to the complexity and the performance, respectively, of the
decoding algorithm. Therefore, we characterize the trade-off between
those two aspects using the relationship between $\bar{R}$ and $\bar{D}$.
In this paper, we denote the rate and distortion by $R$ and $D$,
respectively, using unnormalized quantities, i.e., $R=N\bar{R}$ and
$D=N\bar{D}$.

\subsubsection{RDE approach}

The above-mentioned RD approach focuses on minimizing the average
minimum distortion with little knowledge of how the tail of the distribution
behaves. In this RDE approach, we instead focus on directly minimizing
the probability that the minimum distortion is not less than the predetermined
threshold $D=N-K+1$ (due to the condition (\ref{eq:distorthreshold}))
with the help of an error-exponent analysis. The exact probability
of interest is \[p_{e}=\Pr\left(X^{N}:\min_{\hat{x}^{N}\in\mathcal{B}}d(X^{N},\hat{x}^{N})>D\right)\]
that reflects how likely the decoding threshold (\ref{eq:scbma_0})
is going to fail. In other words, every error pattern $x^{N}$ can
be covered by a sphere centered at an erasure pattern $\hat{x}^{N}$
except for a set of error patterns of probability $p_{e}$. The RDE
analysis shows that $p_{e}$ decays exponentially as $N\to\infty$
and the maximum exponent attainable is the RDE function $F(R,D)$.
Throughout this paper, we denote the rate-distortion exponent by $F(R,D)$
using unnormalized quantities (i.e., without dividing by $N$) and
note that exponent used by other authors in \cite{Blahut-it74,Marton-it74,Csiszar-1981}
is often the normalized version $\bar{F}(R,D)\triangleq\frac{F(R,D)}{N}$. 

RDE analysis is discussed extensively in \cite{Blahut-it74,Marton-it74}
and it is shown that a set $\mathcal{B}$ of roughly
$2^{N\bar{R}}$ codewords, generated randomly using the test-channel
input distribution,
can be used to achieve $\bar{F}(R,D)$.
An upper bound is also given that shows, for any $\epsilon>0$, there is a
sufficiently large $N$ (see \cite[p. 229]{Blahut-1987}) such that
\begin{equation*}
p_{e}\leq2^{-N[\bar{F}(R,D)-\epsilon]}.\end{equation*}
An exponentially tight lower bound for $p_{e}$ can also be obtained
(see \cite[p. 236]{Blahut-1987}) and it implies that the best sequence
of codebooks satisfy
\[ \lim_{N\to\infty}-\frac{1}{N}\log p_{e}=\bar{F}(R,D).\]

\begin{remrk}
\label{rem:advantage}The RDE approach possesses several advantages.
First, the converse of the RDE \cite[p. 236]{Blahut-1987} provides
a lower bound for $p_{e}$. This implies that, given an arbitrary
set $\mathcal{B}$ of roughly $2^{N\bar{R}}$ erasure patterns and
any $\epsilon>0$, the probability $p_{e}$ cannot be made lower than
$2^{-N[\bar{F}(R,D)+\epsilon]}$ for $N$ large enough. Thus, no matter
how one chooses the set $\mathcal{B}$ of erasure patterns, the difference
between the induced probability of error and the $p_{e}$ for the
RDE approach becomes negligible for $N$ large enough. Second, it
can help one estimate the smallest number of decoding attempts to
get to a RDE of $F$ (or get to an error probability of roughly $2^{-N\bar{F}}$)
or, similarly, allow one to estimate the RDE (and error probability)
for a fixed number of decoding attempts.
\end{remrk}

\subsection{Generalized error patterns and erasure patterns}

In this subsection, we consider a generalization of the conventional
error patterns and erasure patterns under the same framework to make
better use of the soft information. At each index of the RS codeword,
besides erasing a symbol, we also try to decode using not only the
most likely symbol but also less likely ones as the hard decision
(HD) symbol. To handle up to the $\ell$ most likely symbols at each
index $i$, we let $\mathbb{Z}_{\ell+1}\triangleq\{0,1,\ldots,\ell\}$
and consider the following definition.
\begin{definitn}
\label{def:(PatternsBMA)}(Generalized error patterns and erasure
patterns) Consider a positive integer $\ell$ smaller than the field
size $m$. Let $x^{N}\in\mathbb{Z}_{\ell+1}^{N}$ be a \emph{generalized
error pattern} where, at index $i$, $x_{i}=j$ implies that the $j$-th
most likely symbol is correct for $j\in\{1,2,\ldots\ell\}$, and $x_{i}=0$
implies none of the first $\ell$ most likely symbols is correct.
Let $\hat{x}^{N}\in\mathbb{Z}_{\ell+1}^{N}$ be a \emph{generalized
erasure pattern} used for decoding where, at index $i$, $\hat{x}_{i}=k$
implies that the $k$-th most likely symbol is used as the hard-decision
symbol for $k\in\{1,2,\ldots,\ell\}$, and $\hat{x}_{i}=0$ implies
that an erasure is used at that index. 

For simplicity, we refer to $x^{N}$ as the error pattern and $\hat{x}^{N}$
as the erasure pattern like in the conventional case. Now, we need
to convert the condition (\ref{eq:scbma_0}) to the form where $d(x^{N},\hat{x}^{N})$
is less than a fixed threshold. Proposition \ref{prop:bma1}\emph{
}is thereby generalized into the following proposition.\end{definitn}
\begin{prop}
\label{thm:ExtBMA-1}We choose the \emph{letter-by-letter} distortion
measure $\delta:\mathcal{X}\times\hat{\mathcal{X}}\rightarrow\mathbb{R}_{\geq0}$,
where in this case $\mathcal{X}=\hat{\mathcal{X}}=\mathbb{Z}_{\ell+1}$,
defined by $\delta(x,\hat{x})=[\Delta]_{x,\hat{x}}$ in terms of the
$(\ell+1)\times(\ell+1$) matrix\begin{equation}
\Delta=\left(\begin{array}{ccccc}
1 & 2 & \ldots & 2 & 2\\
1 & 0 & \ldots & 2 & 2\\
\vdots & \vdots & \ddots & \vdots & \vdots\\
1 & 2 & \ldots & 0 & 2\\
1 & 2 & \ldots & 2 & 0\end{array}\right).\label{eq:dstmBMl}\end{equation}
Using this, the condition (\ref{eq:scbma_0}) for a successful errors-and-erasures
decoding is equivalent to\[
d(x^{N},\hat{x}^{N})<N-K+1.\]
\end{prop}
\begin{proof}
The reasoning is very similar to the proof of Proposition \ref{prop:bma1}
using the fact that $e=\sum_{j=0}^{\ell}\chi_{j,0}$ and $\nu=\sum_{k=1}^{\ell}\sum_{j=0,j\neq k}^{\ell}\chi_{j,k}$
where $\chi_{j,k}\triangleq\left|\{i\in\left\{ 1,2,\ldots,N\right\} :x_{i}=j,\hat{x}_{i}=k\}\right|$
for every $j,k\in\mathbb{Z}_{\ell+1}$.
\end{proof}
For each $\ell=1,2,\ldots,m$, we will refer to this generalized case
as mBM-$\ell$ decoding.
\begin{example}
Consider mBM-2 (or top-$\ell$ decoding with $\ell=2$). In this case,
the distortion measure is given by following the matrix\[
\Delta=\left(\begin{array}{ccc}
1 & 2 & 2\\
1 & 0 & 2\\
1 & 2 & 0\end{array}\right).\]
\end{example}
\begin{remrk}
The distortion measure matrix changes slightly if we use the errors-only
decoding instead of errors-and-erasures decoding. In this case, $\hat{\mathcal{X}}=\mathbb{Z}_{\ell+1}\setminus\{0\}$
and the chosen letter-by-letter distortion measure is given in terms
of the $(\ell+1)\times\ell$ matrix obtained by deleting the first
column of (\ref{eq:dstmBMl}). When $\ell=2,$ we consider the first
and second most likely symbols as the two hard-decision symbols at
each codeword position. This is similar to the Chase-type decoding
method proposed by Bellorado and Kavcic \cite{Bellorado-it10}. Das
and Vardy also suggest this approach by considering only several highest
entries in each column of the reliability matrix $\Pi$ for single
ASD decoding of RS codes \cite{Das-isit09}.
\end{remrk}

\subsection{Proposed General Multiple-Decoding Algorithm \label{sec:Proposed-Algorithm}}

In this section, we propose two general multiple-decoding algorithms
for RS codes. In each algorithm, one can choose either Step 2a that
corresponds to the RD approach or Step 2b that corresponds to the
RDE approach. These general algorithms apply to not only multiple
errors-and-erasures decoding but also multiple-decoding of other decoding
schemes that we will discuss later. The common first step is designing
a distortion measure $\delta:\mathcal{X}\times\hat{\mathcal{X}}\rightarrow\mathbb{R}_{\geq0}$
that converts the condition for a single decoding to succeed to the
form where distortion is less than a fixed threshold. After that,
decoding proceeds as described below.

\subsubsection{Algorithm A}

~

\emph{Step 1:} Based on the received signal sequence, compute an $m\times N$
reliability matrix $\mathbf{\Pi}$ where $[\mathbf{\Pi}]_{j,i}=\pi_{i,j}$.
From this, determine the probability matrix $\mathbf{P}$ where $p_{i,j}=\Pr(X_{i}=j)$
for $i=1,2,\ldots,N$ and $j\in\mathcal{X}.$

\emph{Step 2a:} (RD approach) Compute the RD function of a source
sequence (error pattern) with probability of source letters derived
from $\mathbf{P}$ and the chosen distortion measure (see Section
\ref{sec:Computing-RD} and Section \ref{sec:Closed-Form-Analysis-of}).
Given the design rate $R$, determine the optimal input-probability
distribution matrix $\mathbf{Q}$, for the test channel, with entries
$q_{i,k}=\Pr(\hat{X}_{i}=k)$ for $i=1,2,\ldots,N$ and $k\in\hat{\mathcal{X}}.$

\emph{Step 2b:} (RDE approach) Given $D$ (in most cases $D=N-K+1)$
and the design rate $R$, compute the RDE function of a source
sequence (error pattern) with probability of source letters derived
from $\mathbf{P}$ and the chosen distortion measure (see Section
\ref{sec:Computing-RD} and Section \ref{sec:Closed-Form-Analysis-of}).
Also determine the optimal input-probability distribution matrix $\mathbf{Q}$,
for the test channel, with entries $q_{i,k}=\Pr(\hat{X}_{i}=k)$ for
$i=1,2,\ldots,N$ and $k\in\hat{\mathcal{X}}.$

\emph{Step 3:} Randomly generate a set of $2^{R}$ erasure patterns
using the test-channel input-probability distribution matrix $\mathbf{Q}$. 

\emph{Step 4:} Run multiple attempts of the corresponding decoding
scheme (e.g., errors-and erasures decoding) using the set of erasure
patterns in Step 3 to produce a list of candidate codewords.

\emph{Step 5:} Use the maximum-likelihood (ML) rule to pick the best
codeword on the list.
\begin{remrk}
In Algorithm A, the RD (or RDE) function is computed on the fly, i.e.,
after every received signal sequence. In practice, it may be preferable
to precompute the RD (or RDE) function based on the empirical distribution
measured from the channel. We refer to this approach as Algorithm
B, and simulation results show a negligible difference in the performance
of these two algorithms.
\end{remrk}

\subsubsection{Algorithm B}

~%

\emph{Step 1:} Transmit $\tau$ (e.g., $\tau=10^{3}-10^{6}$) arbitrary
test RS codewords, indexed by time $t=1,2,\ldots,\tau$, over the
channel and compute a set of $\tau$ $m\times N$ matrices $\mathbf{\Pi}_{1}^{(t)}$
where $[\mathbf{\Pi}_{1}^{(t)}]_{j,i}=\pi_{i,\varphi_{i}^{(t)}(j)}^{(t)}$
is the probability of the $j$-th most likely symbol at position $i$
during time $t$. For each time~$t$, obtain the matrix $\mathbf{\Pi}_{2}^{(t)}$
from $\mathbf{\Pi}_{1}^{(t)}$ through a permutation $\sigma^{(t)}:\{1,2,\ldots,N\}\rightarrow\{1,2,\dots,N\}$
that sorts the probabilities $\pi_{i,\varphi_{i}^{(t)}(1)}^{(t)}$
in increasing order to indicate the reliability order of codeword
positions. Take the entry-wise average of all $\tau$ matrices $\mathbf{\Pi}_{2}^{(t)}$
to get an average matrix $\mathbf{\bar{\Pi}}$.%
\footnote{In fact, one need not store separately each $\Pi_{2}^{(t)}$ matrix.
The average $\bar{\Pi}$ can be computed on the fly.%
} The matrix $\mathbf{\bar{\Pi}}$ serves as $\mathbf{\Pi}$ in Algorithm
A and from this, determine the probability matrix $\mathbf{P}$ where
$p_{i,j}=\Pr(X_{i}=j)$ for $i=1,2,\ldots,N$ and $j\in\mathcal{X}.$

\emph{Step 2a:} (RD approach) Compute the RD function of a source
sequence (error pattern) with probability of source letters derived
from \textbf{$\mathbf{P}$} and the chosen distortion measure. Given
a design rate $R$, determine the test-channel input-probability
distribution matrix $\mathbf{Q}$ where $q_{i,k}=\Pr(\hat{X}_{i}=k)$
for $i=1,2,\ldots,N$ and $k\in\hat{\mathcal{X}}.$

\emph{Step 2b:} (RDE approach) Given $D$ (in most cases $D=N-K+1)$
and the design rate $R$, compute the RDE function of a source
sequence (error pattern) with probability of source letters derived
from $\mathbf{P}$ and the chosen distortion measure. Also determine
the optimal test-channel input-probability distribution matrix $\mathbf{Q}$
where $q_{i,k}=\Pr(\hat{X}_{i}=k)$ for $i=1,2,\ldots,N$ and $k\in\hat{\mathcal{X}}.$

\emph{Step 3: }Based on the actual received signal sequence, compute
$\pi_{i,\varphi_{i}(1)}$ and determine the permutation~$\sigma$~that
gives the reliability order of codeword positions by sorting $\pi_{i,\varphi_{i}(1)}$
in increasing order.

\emph{Step 4:} Randomly generate a set of $2^{R}$ erasure patterns
using the test-channel input-probability distribution matrix $\mathbf{Q}$
and permute the indices of each erasure pattern by the permutation
$\sigma^{-1}.$

\emph{Step 5:} Run multiple attempts of the corresponding decoding
scheme (e.g., errors-and-erasures decoding) using the set of erasure
patterns in Step 4 to produce a list of candidate codewords.

\emph{Step 6:} Use the ML rule to pick the best codeword on the list.

\section{Computing The RD and RDE Functions\label{sec:Computing-RD}}

In this section, we will discuss some numerical methods to compute
the RD and RDE functions and the corresponding test-channel input-probability
distribution matrix $\mathbf{Q}$, whose entries are $q_{i,k}=\Pr(\hat{X}_{i}=k)$
for $i=1,2,\ldots,N$ and $k\in\hat{\mathcal{X}}$. These numerical
methods allow us to efficiently compute the RD and RDE functions discussed
in the previous section for arbitrary discrete distortion measures.
For some simple distortion measures, closed-form solutions are given
in Section \ref{sec:Closed-Form-Analysis-of}.

\subsection{Computing the RD function\label{sub:Computing-the-R-D}}

For an arbitrary discrete distortion measure, it can be difficult
to compute the RD function analytically. Fortunately, for a single
source $X$, the Blahut algorithm (see details in \cite{Blahut-it72})
gives an alternating minimization technique that efficiently computes
the RD function which is given by%
\footnote{All logarithms in this paper are taken to base 2.%
}\[
R(D)=\min_{\mathbf{w}\in\mathcal{W}_{D}}\sum_{j}\sum_{k}p_{j}w_{k|j}\log\frac{w_{k|j}}{\sum_{j'}p_{j'}w_{k|j'}}\]
where $p_{j}\triangleq\Pr(X=j)$, $q_{k}\triangleq\Pr(\hat{X}=k)$,
$w_{k|j}\triangleq\Pr(\hat{X}=k|X=j)$, and%
\footnote{$\delta(j,k)$ is sometimes written as $\delta_{jk}$ for convenience.%
}

\[
\mathcal{W}_{D}=\left\{ \mathbf{w}\bigg|\begin{array}{c}
w_{k|j}\geq0,\sum_{k}w_{k|j}=1\\
\sum_{j}\sum_{k}p_{j}w_{k|j}\delta_{jk}\leq D\end{array}\right\} .\]
More precisely, given the Lagrange multiplier $t\leq0$
that represents the slope of the RD curve at a specific point
(see \cite[Thm 2.5.1]{Berger-1971}) and an arbitrary all-positive
initial test-channel input-probability distribution vector $\underbar{q}^{(0)}$,
the Blahut algorithm shows us how to compute the rate-distortion pair
$(R_{t},D_{t})$. 

However, it is not straightforward to apply the Blahut algorithm to
compute the RD for a discrete source sequence $x^{N}$ (an error pattern
in our context) of $N$ independent but not necessarily identical
(i.n.d.) source components $x_{i}$. In order to do that, we consider
the group of source letters $(j_{1},j_{2},\ldots,j_{N})$ where $j_{i}\in\mathcal{X}$
as a super-source letter $\mathcal{J}\in\mathcal{X}^{N}$, the group
of reproduction letters $(k_{1},k_{2},\ldots,k_{N})$ where $k_{i}\in\mathcal{\hat{X}}$
as a super-reproduction letter $\mathcal{K}\in\hat{\mathcal{X}}^{N}$,
and the source sequence $x^{N}$ as a single source. For each super-source
letter $\mathcal{J}$, $p_{\mathcal{J}}=\Pr(X^{N}=\mathcal{J})=\prod_{i=1}^{N}\Pr(X_{i}=j_{i})=\prod_{i=1}^{N}p_{j_{i}}$
follows from the independence of source components.%
\footnote{In this paper, the notations $p_{j_{i}}$ and $p_{i,j}$ are interchangeable.
The notations $q_{k_{i}}$ and $q_{i,k}$ are also interchangeable.}

While we could apply the Blahut algorithm to this source directly,
the complexity is a problem because the alphabet sizes for $\mathcal{J}$
and $\mathcal{K}$ become the super-alphabet sizes $|\chi|^{N}$ and
$|\hat{\chi}|^{N}$ respectively.  Instead, we avoid this computational
challenge by choosing the initial test-channel input-probability distribution
so that it can be factored into a product of $N$ initial test-channel
input-probability components, i.e., $q_{\mathcal{K}}^{(0)}=\prod_{i=1}q_{k_{i}}^{(0)}$.
One can verify that this factorization rule still applies after every
step $\tau$ of the iterative process, i.e., $q_{\mathcal{K}}^{(\tau)}=\prod_{i=1}q_{k_{i}}^{(\tau)}$.
Therefore, the convergence of the Blahut algorithm \cite{Csiszar-it74}
implies that the optimal distribution is a product distribution, i.e.,
$q_{\mathcal{K}}^{\star}=\prod_{i=1}q_{k_{i}}^{\star}$.

One can also finds that, for each parameter $t$, one only needs to compute the
rate-distortion pair for each source component $x_{i}$ separately
and sum them together. This is captured into the following algorithm.

\begin{alg}\label{thm:(Factored-Blahut)}(Factored Blahut algorithm
for RD function) Consider a discrete source sequence $x^{N}$ of $N$
i.n.d. source components $x_{i}$'s with probability $p_{j_{i}}\triangleq\Pr(X_{i}=j_{i})$.
Given a parameter $t\leq0$, the rate and the distortion for this
source sequence under a specified distortion measure are given by\begin{equation}
R_{t}=\sum_{i=1}^{N}R_{i,t}\,\,\text{and}\,\, D_{t}=\sum_{i=1}^{N}D_{i,t}\label{eq:RDit}\end{equation}
where the components $R_{i,t}$ and $D_{i,t}$ are computed by the
Blahut algorithm with the Lagrange multiplier $t$. This rate-distortion
pair can be achieved by the corresponding test-channel input-probability
distribution $q_{\mathcal{K}}\triangleq\Pr(\hat{X}^{N}=\mathcal{K})=\prod_{i=1}^{n}q_{k_{i}}$
where the component probability distribution $q_{k_{i}}\triangleq\Pr(\hat{X}_{i}=k_{i})$.\end{alg}
\begin{remrk}
Equation (\ref{eq:RDit}) can also be derived from \cite[Corollary 2.8.3]{Berger-1971}
in a way that does not use the convergence property of the Blahut
algorithm.
\end{remrk}

\subsection{Computing the RDE function\label{sub:Computing-the-RDE}}

The original RDE function $F(R,D)$, defined in \cite[Sec. VI]{Blahut-it74}
for a single source $X$, is given by\begin{equation}
F(R,D)=\max_{\mathbf{w}}\min_{\mathbf{\tilde{p}}\in\mathcal{P}_{R,D}}\sum_{j}\tilde{p}_{j}\log\frac{\tilde{p}_{j}}{p_{j}}\label{eq:RDEmaxmin}\end{equation}
where $p_{j}=\Pr(X=j)$, $q_{k}=\Pr(\hat{X}=k)$, $w_{k|j}=\Pr(\hat{X}=k|X=j)$,
and

\begin{equation}
\mathcal{P}_{R,D}=\left\{ \mathbf{\tilde{p}}\bigg|\begin{array}{c}
\sum_{j}\sum_{k}\tilde{p}_{j}w_{k|j}\log\frac{w_{k|j}}{\sum_{j'}\tilde{p}_{j'}w_{k|j'}}\geq R\\
\sum_{j}\sum_{k}\tilde{p}_{j}w_{k|j}\delta_{jk}\geq D\end{array}\right\} .\label{eq:RDEmaxminPRD}\end{equation}

For a single source $X$, given two parameters $s\geq0$ and $t\leq0$
which are the Lagrange multipliers introduced in the optimization
problem (see \cite[p. 415]{Blahut-it74}), the Arimoto algorithm given
in \cite[Sec. V]{Arimoto-it76} can be used to compute the exponent,
rate, and distortion numerically.

In the context we consider, the source (error pattern) $x^{N}$ comprises
i.n.d. source components $x_{i}$'s. We follow the same method as
in the RD function case, i.e., by choosing the initial distribution
still arbitrarily but following a factorization rule $q_{\mathcal{K}}^{(0)}=\prod_{i=1}^{N}q_{k_{i}}^{(0)}$,
and this gives the following algorithm. 

\begin{alg}\label{pro:FactoredRDE}(Factored Arimoto algorithm for
RDE function) Consider a discrete source $x^{N}$ of i.n.d. source
components $x_{i}$'s with probability $p_{j_{i}}\triangleq\Pr(X_{i}=j_{i})$.
Given Lagrange multipliers $s\geq0$ and $t\leq0$, the exponent,
rate and distortion under a specified distortion measure are given
by\[ \left.F\right|_{s,t}=\sum_{i=1}^{N}\left.F_i\right|_{s,t},\,\, \left.R\right|_{s,t}=\sum_{i=1}^{N}\left.R_i\right|_{s,t},\,\, \left.D\right|_{s,t}=\sum_{i=1}^{N}\left.D_i\right|_{s,t}\]where
the components $\left.F_i\right|_{s,t},\left.R_i\right|_{s,t},\left.D_i\right|_{s,t}$
are computed parametrically by the Arimoto algorithm.\end{alg}
\begin{remrk}
Though it is standard practice to compute error-exponents using the
implicit form given above, this approach may provide points that,
while achievable, are strictly below the true RDE curve. The problem
is that the true RDE curve may have a slope discontinuity that forces
the implicit representation to have extra points. An example of this
behavior for the channel coding error exponent is given by Gallager
\cite[p. 147]{Gallager-1968}. For the i.n.d. source considered above,
a cautious person could solve the problem as described and then check
that the component RDE functions are differentiable at the optimum
point. In this work, we largely neglect this subtlety.
\end{remrk}

\subsection{Complexity of computing RD/RDE functions}

\subsubsection{Complexity of computing RD function.}

For each parameter $t<0,$ if we directly apply of the original Blahut
algorithm to compute the $(R_{t},D_{t})$ pair, the complexity is
$O(\mbox{\ensuremath{\tau}}_{\max}|\mathcal{X}|^{N}|\hat{\mathcal{X}}|^{N})$
where $\tau_{\max}$ is the number of iterations in the Blahut algorithm.
However, using the factored Blahut algorithm (Algorithm \ref{thm:(Factored-Blahut)})
greatly reduces this complexity to $O(\tau_{\max}|\mathcal{X}||\hat{\mathcal{X}}|N)$.
In Section~\ref{sec:Proposed-Algorithm}, one of the proposed algorithms
needs to compute the RD function for a design rate $R$. To do this,
we apply the bisection method on $t$ to find the correct $t$ that
corresponds to the chosen rate $R$. 
\begin{itemize}
\item \emph{Step~0}: Set $t_{\min}<0$ (e.g., $t_{\min}=-10$)
\item \emph{Step~1}: If $R_{t_{\min}}>R,$ go to Step 3. Else go to Step 2.
\item \emph{Step~2}: If $R_{t_{\min}}=R$ then stop. Else if $R_{t_{\min}}<R$,
set $t_{\min}\leftarrow2t_{\min}$ and go to Step 1.
\item \emph{Step~3}: Find $t$ using the bisection method to get the correct rate
$R$ within $\epsilon_{R}$.
\end{itemize}
The overall complexity of computing the RD function for a design
rate $R$ is \[O\left(\tau_{\max}\log_{2}\left(\frac{-t_{\min}}{\epsilon_{R}}\right)|\mathcal{X}||\hat{\mathcal{X}}|N\right).\] 

Now, we consider the dependence of $\tau_{\max}$ on $\epsilon_{R}$.
It follows from \cite{Csiszar-it74} that the error due to early termination
of the Blahut algorithm is $O\left(\frac{1}{\tau_{\max}}\right)$.
This implies that choosing $\tau_{\max}=O\left(\frac{1}{\epsilon_{R}}\right)$
is sufficient. However, recent work has shown that a slight modification
of the Blahut algorithm can drastically increase the convergence rate
\cite{Matz-itw04}. For this reason, we leave the number of iterations
as the separate constant $\tau_{\max}$ and do not consider its relationship
to the error tolerance.

\subsubsection{Complexity of computing RDE function.}

Similarly, for each pair of parameters $t<0$ and $s\geq0$, the complexity
if we directly apply of the original Arimoto algorithm to compute
the $(R|_{s,t},D|_{s,t})$ pair is $O(\tau_{\max}|\mathcal{X}|^{N}|\hat{\mathcal{X}}|^{N})$
where $\tau_{\max}$ is the number of iterations. Instead, if the
factored Arimoto algorithm (Algorithm \ref{pro:FactoredRDE}) is employed,
this complexity can also be reduced to $O(\tau_{\max}|\mathcal{X}||\hat{\mathcal{X}}|N)$.
In one of our proposed general algorithms in Section \ref{sec:Proposed-Algorithm},
we need to compute the RDE function for a pre-determined $(R,D)$
pair. We use a nested bisection technique to find the Lagrange multipliers
$s,t$ that give the correct $R$ and $D$.
\begin{itemize}
\item \emph{Step~0}: Set $t_{\min}<0$ and $s_{\max}>0$ $($e.g., $t_{\min}=-10$
and $s_{\max}=2$)
\item \emph{Step~1}: If $R|_{s_{\max},t_{\min}}\leq R$, set $t_{\min}\leftarrow2t_{\min}$
and repeat Step 1. Else go to Step 2.
\item \emph{Step~2}: Find $t$ using the bisection method to obtain $R|_{s_{\max},t}=R$
within $\epsilon_{R}$. If $D|_{s_{\max},t}>D,$ go to Step 3. If
$D|_{s_{\max},t}=D$ then stop. Else if $D|_{s_{\max},t}<D$, set
$s_{\max}\leftarrow2s_{\max}$ and go to Step 1.
\item \emph{Step~3}: Find $s$ using the bisection method to get the correct distortion
$D$ within $\epsilon_{D}$ while with each $s$ doing the following
steps

\begin{itemize}
\item \emph{Step~3a}: If $R|_{s,t_{\min}}>R$, go to Step 3c.
\item \emph{Step~3b}: If $R|_{s,t_{\min}}=R$, then stop. Else if $R|_{s,t_{\min}}<R$,
set $t_{\min}\leftarrow2t_{\min}$ and go to Step 1. 
\item \emph{Step~3c}: Find $t$ using the bisection method to get the correct $R$
within $\epsilon_{R}$.
\end{itemize}
\end{itemize}
The overall complexity of computing the RD function for a design
rate $R$ is therefore \[
O\left(\tau_{\max}\log_{2}\left(\frac{-t_{\min}}{\epsilon_{R}}\right)\log_{2}\left(\frac{s_{\max}}{\epsilon_{D}}\right)|\mathcal{X}||\hat{\mathcal{X}}|N\right).\]

\section{Multiple Algebraic Soft-Decision (ASD) Decoding \label{sec:Multiple ASD}}

In this section, we analyze and design a distortion measure to convert
the condition for successful ASD decoding to a suitable form so that
we can apply the general multiple-decoding algorithm to ASD decoding.

First, let us give a brief review on ASD decoding of RS codes. Let
$\{\beta_{1},\beta_{2},\ldots,\beta_{N}\}$ be a set of $N$ distinct
elements in $\mathbb{F}_{m}.$ From each message polynomial $f(X)=f_{0}+f_{1}X+\ldots+f_{K-1}X^{K-1}$
whose coefficients are in $\mathbb{F}_{m}$, we can obtain a codeword
$c=(c_{1},c_{2},\ldots,c_{N})$ by evaluating the message polynomial
at $\{\beta_{i}\}_{i=1}^{N}$, i.e., $c_{i}=f(\beta_{i})$ for $i=1,2,\ldots,N$.
Given a received vector $\mathbf{r}=(r_{1},r_{2},\ldots,r_{N})$,
we can compute the \emph{a posteriori} probability (APP) matrix $\mathbf{\Pi}$
as follows:\[
[\mathbf{\Pi}]_{j,i}=\pi_{i,j}=\Pr(c_{i}=\alpha_{j}|r_{i})\,\,\,\mbox{for}\,\,1\leq i\leq N,1\leq j\leq m.\]
The ASD decoding as in \cite{Koetter-it03} has the following main
steps.
\begin{enumerate}
\item \emph{Multiplicity Assignment}: Use a particular multiplicity assignment
scheme (MAS) to derive an $m\times N$ multiplicity matrix, denoted
as $\mathbf{M}$, of non-negative integer entries $\{M_{i,j}\}$ from
the APP matrix $\mathbf{\Pi}$.
\item \emph{Interpolation}: Construct a bivariate polynomial $Q(X,Y)$ of
minimum $(1,K-1)$ weighted degree that passes through each of the
point $(\beta_{j},\alpha_{i})$ with multiplicity $M_{i,j}$ for $i=1,2,\ldots,m$
and $j=1,2,\ldots,N$.
\item \emph{Factorization}: Find all polynomials $f(X)$ of degree less
than $K$ such that $Y-f(X)$ is a factor of $Q(X,Y)$ and re-evaluate
these polynomials to form a list of candidate codewords.
\end{enumerate}
In this paper, we denote $\mu=\max_{i,j}M_{i,j}$ as the maximum multiplicity.
Intuitively, higher multiplicity should be put on more likely symbols.
A higher $\mu$ generally allows ASD decoding to achieve a better
performance. However, one of the drawbacks of ASD decoding is that
its decoding complexity is roughly $O(N^{2}\mu^{4})$ \cite{McEliece-IPN03}.
Even though there have been several reduced complexity variations
and fast architectures as discussed in \cite{Gross-com06,Zhang-vlsi06,Ma-vlsi07},
the decoding complexity still increases rapidly with $\mu$. Thus,
in this section we will mainly work with small $\mu$ to keep the
complexity affordable. 

One of the main contributions of \cite{Koetter-it03} is to offer
a condition for successful ASD decoding represented in terms of two
quantities specified as the score and the cost as follows.
\begin{definitn}
The score $S_{\mathbf{M}}(\mathbf{c})$ with respect to a codeword
$\mathbf{c}$ and a multiplicity matrix $\mathbf{M}$ is defined as\[
S_{\mathbf{M}}(\mathbf{c})=\sum_{j=1}^{N}M_{[c_{j}],j}\]
where $[c_{j}]=i$ such that $\alpha_{i}=c_{j}$. The cost $C_{\mathfrak{\mathbf{M}}}$
of a multiplicity matrix $\mathbf{M}$ is defined as \[C_{\mathbf{M}}=\frac{1}{2}\sum_{i=1}^{m}\sum_{j=1}^{N}M_{i,j}(M_{i,j}+1).\]\end{definitn}
\begin{condition}
\label{con:ASDcond}(ASD decoding threshold, see \cite{Koetter-it03,Jiang-it08,McEliece-IPN03}).
The transmitted codeword will be on the list if\begin{equation}
(a+1)\left[S_{\mathbf{M}}-\frac{a}{2}(K-1)\right]>C_{\mathbf{M}}\label{eq:maincond}\end{equation}
for some $a\in\mathbb{N}$ such that
\begin{equation}
a(K-1)<S_{\mathbf{M}}\leq(a+1)(K-1).\label{eq:piecewise}\end{equation}

To match the general framework, the ASD decoding threshold (or condition
for successful ASD decoding) should be converted to the form where
the distortion is smaller than a fixed threshold. 
\end{condition}

\subsection{Bit-level ASD case}
In this subsection, we consider multiple trials of ASD decoding using
bit-level erasure patterns. A bit-level error pattern $b^{n}\in\mathbb{Z}_{2}^{n}$
and a bit-level erasure pattern $\hat{b}^{n}\in\mathbb{Z}_{2}^{n}$
have length $n=N\times\eta$ since each symbol has $\eta$ bits. Similar
to Definition \ref{def:(Conv. patterns)} of a conventional error
pattern and a conventional erasure pattern, $b_{i}=0$ in a bit-level
error pattern implies a bit-level error occurs and $\hat{b}_{i}$
in a bit-level erasure pattern implies that a bit-level erasure is
applied. We also use $B^{N}$ and $\hat{B}^{N}$ to denote the random
vectors which generate the realizations $b^{N}$ and $\hat{b}^{N}$,
respectively.

From each bit-level erasure pattern, we can specify entries of the
multiplicity matrix $\mathbf{M}$ using the bit-level MAS proposed
in \cite{Jiang-it08} as follows: for each codeword position, assign
multiplicity 2 to the symbol with no bit erased, assign multiplicity
1 to each of the two candidate symbols if there is 1 bit erased, and
assign multiplicity zero to all the symbols if there are $\geq2$
bits erased. All the other entries are zeros by default. This MAS
has a larger decoding region compared to the conventional errors-and-erasures
decoding scheme.
\begin{condition}
(Bit-level ASD decoding threshold, see \cite{Jiang-it08}) For RS
codes of rate $\frac{K}{N}\geq\frac{2}{3}+\frac{1}{N}$, ASD decoding
using the bit-level MAS will succeed (i.e., the transmitted codeword
is on the list) if\begin{equation}
3\nu_{b}+e_{b}<\frac{3}{2}(N-K+1)\label{eq:bgmdsc}\end{equation}
where $e_{b}$ is the number of bit-level erasures and $\nu_{b}$
is the number of bit-level errors in unerased locations.
\end{condition}
We can choose an appropriate distortion measure according to the following
proposition which is a natural extension of Proposition \ref{prop:bma1}
in the symbol level.
\begin{prop}
\label{prop:BitASD}If we choose the bit-level \emph{letter-by-letter}
distortion measure~$\delta:\mathbb{Z}_{2}\times\mathbb{Z}_{2}\rightarrow\mathbb{R}_{\geq0}$~as
follows\begin{equation*}
\begin{array}{cc}
\delta(0,0)=1, & \delta(0,1)=3,\\
\delta(1,0)=1, & \delta(1,1)=0,\end{array}\label{eq:dstfnBGMD}\end{equation*}
then the condition (\ref{eq:bgmdsc}) becomes \begin{equation}
d(b^{n},\hat{b}^{n})<\frac{3}{2}\left(N-K+1\right).\end{equation}
\end{prop}
\begin{proof}
The condition (\ref{eq:bgmdsc}) can be seen to be equivalent to \[
\frac{2}{3}d(b^{n},\hat{b}^{n})<N-K+1\]
using the same reasoning as in\emph{ }Proposition\emph{ }\ref{prop:bma1}.
The results then follows right away.\end{proof}
\begin{remrk}
We refer the multiple-decoding of bit-level ASD as m-bASD.
\end{remrk}

\subsection{Symbol-level ASD case}

In this subsection, we try to convert the condition for successful
ASD decoding in general to the form that suits our goal. We will also
determine which multiplicity assignment schemes allow us to do so. 
\begin{definitn}
(Multiplicity type) Consider a positive integer $\ell\leq m$ where
$m$ is the number of elements in $\mathbb{F}^{m}$. For some codeword
position, let us assign multiplicity $m_{j}$ to the $j$-th most
likely symbol for $j=1,2,\ldots,\ell$. The remaining entries in the
column are zeros by default. We call the sequence, $(m_{1},m_{2},\ldots,m_{\ell})$,
the column \emph{multiplicity type} for {}``top-$\ell$'' decoding.
\end{definitn}
First, we notice that a choice of multiplicity types in ASD decoding
at each codeword position has the similar meaning to a choice of erasure
decisions in the conventional errors-and-erasures decoding. However,
in ASD decoding we are more flexible and may have more types of erasures.
For example, assigning multiplicity zero to all the symbols (all-zero
multiplicity type) at codeword position $i$ is similar to erasing
that position. Assigning the maximum multiplicity $\mu$ to one symbol
corresponds to the case when we choose that symbol as the hard-decision
one. Hence, with some abuse of terminology, we also use the term (generalized)
erasure pattern $\hat{x}^{N}$ for the multiplicity assignment scheme
in the ASD context. Each erasure-letter $x_{i}$ gives the multiplicity
type for the corresponding column of the multiplicity matrix $\mathbf{M}$. 
\begin{definitn}
(Error patterns and erasure patterns for ASD decoding) Consider a
MAS with $T$ multiplicity types. Let $\hat{x}^{N}\in\{1,2\ldots,T\}^{N}$
be an erasure pattern where, at index $i$, $x_{i}=j$ implies that
multiplicity type $j$ is used at column $i$ of the multiplicity
matrix $\mathbf{M}$. Notice that the definition of an error pattern
$x^{N}\in\mathbb{Z}_{\ell+1}^{N}$ in Definition \ref{def:(PatternsBMA)}
applies unchanged here. \end{definitn}
\begin{remrk}
In our method, we generally choose an appropriate integer $a$ in
Condition \ref{con:ASDcond} and design a distortion measure corresponding
to the chosen $a$ so that the condition for successful ASD decoding
can be converted to the form where distortion is less than a fixed
threshold. The following definition of allowable multiplicity types
will lead us to the result of Lemma \ref{lem:GenASD} and consequently,
$a\geq\mu$, as stated in Corollary \ref{cor:ageqmu}. Also, we want
to find as many as possible multiplicity types since rate-distortion
theory gives us the intuition that in general the more multiplicity
types (erasure choices) we have, the better performance of multiple
ASD decoding we achieve as $N$ becomes large.\end{remrk}

\begin{definitn}
\label{def:AllowableTypes}The set of allowable multiplicity types
for {}``top-$\ell$'' decoding with maximum multiplicity $\mu$
is defined to be\footnote{We use the convention that~$\min_{j:m_{j}\neq0}m_{j}=0$ if $\left\{ j:m_{j}\neq0\right\} =\emptyset$.%
}
\begin{equation}
\mathcal{A}(\mu,\ell)\triangleq\left\{ (m_{1},m_{2},\ldots,m_{\ell})\Bigg|\begin{array}{c}
\sum_{j=1}^{\ell}m_{j}\leq\mu,\\
\sum_{j=1}^{\ell}m_{j}(\mu-m_{j})\leq(\mu+1)\left(\left|\left\{ j:m_{j}\neq0\right\} \right|-1\right)\min_{j:m_{j}\neq0}m_{j}\end{array}\right\} .\label{eq:AllowableMAS}
\end{equation}
We take the elements of this set in an arbitrary order and label them
as $1,2,\ldots,|\mathcal{A}(\mu,\ell)|$ with the convention that
the multiplicity type 1 is always $(\mu,0,\ldots,0)$ which assigns
the whole multiplicity $\mu$ to the most likely symbol. The multiplicity
type $k$ is denoted as $(m_{1,k},m_{2,k},\ldots m_{\ell,k})$. \end{definitn}
\begin{remrk}
Multiplicity types $(0,0,\ldots,0),(1,1\ldots,1)$ as well as any
permutations of $(\mu,0,\ldots,0)$ and $(\lfloor\frac{\mu}{2}\rfloor,\lfloor\frac{\mu}{2}\rfloor,0,\ldots,0)$
are always in the allowable set $\mathcal{A}(\mu,\mu)$. We use mASD-$\mu$
to denote the proposed multiple ASD decoding using $\mathcal{A}(\mu,\mu)$.\end{remrk}


\begin{example}
\label{exa:mASD2a}Consider mASD-2 where $\mu=\ell=2$. We have $\mathcal{A}(2,2)\!=\!\{\!(2,0),(1,1),(0,2),(0,0)\}$ which comprises
four allowable multiplicity types for {}``top-2'' decoding as follows:
the first is $(2,0)$ where we assign multiplicity 2 to the most likely
symbol $y_{i,1}$, the second is $(1,1)$ where we assign equal multiplicity
1 to the first and second most likely symbols $y_{i,1}$ and $y_{i,2}$,
the third is $(0,2)$ where we assign multiplicity 2 to the second
most likely symbol $y_{i,2}$, and the fourth is $(0,0)$ where we
assign multiplicity zero to all the symbols at index $i$ (i.e., the
$i$-th column of $\mathbf{M}$ is an all-zero column). We also consider
a restricted set, called mASD-2a, that uses the set of multiplicity
types $\{(2,0),(1,1),(0,0)\}$.
\end{example}

\begin{example}
Consider mASD-3. In this case, the allowable set $\mathcal{A}(3,3)$
consists of all the permutations of $(3,0,0),(0,0,0),(1,1,0),(2,1,0),(1,1,1)$.
We can see that the set $\mathcal{A}(3,2)$ consists of all permutations
of $(3,0),(2,1),(1,1),(0,0)$ and $\left|\mathcal{A}(3,2)\right|<\left|\mathcal{A}(3,3)\right|$.
\end{example}
From now on, we assume that only allowable multiplicity types are
considered throughout most of the paper. With that setting in mind,
we can obtain the following lemmas and theorems.
\begin{lemma}
\label{lem:GenASD}Consider a MAS($\mu,\ell$) for {}``top-$\ell$''
ASD decoding with multiplicity matrix $\mathbf{M}$ that only uses multiplicity
types in the allowable set $\mathcal{A}(\mu,\ell)$. Then, the score and the cost satisfy the
following inequality: \[2C_{\mathbf{M}}\geq(\mu+1)S_{\mathbf{M}}.\]\end{lemma}
\begin{proof}
Let us denote $e_{k}=|\{i\in\{1,\ldots,N\}:\hat{x}_{i}=k\}|$ to count
the number of positions $i$ that use multiplicity type $k$ for $k=1,\ldots,T$
and notice that $\sum_{k=1}^{T}e_{k}=N$. We also use $\nu_{j,k}=|\{i\in\{1,\ldots,N\}:x_{i}\neq j,\hat{x}_{i}=k\}|$
to count the number of positions $i$ that use multiplicity type $k$
where the $j$-th most reliable symbol $y_{i,j}$ is incorrect for
$j=0,\ldots,\ell$ and $k=1,\ldots,T$. The notation $\chi_{j,k}=|\{i\in\{1,\ldots,N\}:x_{i}=j,\hat{x}_{i}=k\}|$
remains the same. Notice also that \begin{equation}e_{k}=\sum_{j=0}^{\ell}\chi_{j,k}\quad
\text{and}\quad\chi_{j,k}=e_{k}-\nu_{j,k}.\label{eqlm1_4}\end{equation}

The score and the cost can therefore be written as
\begin{align}
&S_{\mathbf{M}}(\mathbf{c}) = \sum_{j=1}^{N}M_{[c_{j}],j}\nonumber\\
& \phantom{S_{\mathbf{M}}(\mathbf{c})} =\sum_{k=1}^{T}\sum_{j=1}^{\ell}m_{j,k}\chi_{j,k}\label{eqscorechi}\\
& \phantom{S_{\mathbf{M}}(\mathbf{c})}=\mu\chi_{1,1}+\sum_{k=2}^{T}\sum_{j=1}^{\ell}m_{j,k}\chi_{j,k}\label{eqmult1} \\
& \phantom{S_{\mathbf{M}}(\mathbf{c})}= \mu\left(N-\sum_{k=2}^{T}e_{k}-\nu_{1,1}\right)+\sum_{k=2}^{T}\sum_{j=1}^{\ell}m_{j,k}(e_{k}-\nu_{j,k})
\label{eq:score} 
\end{align}
and
\begin{align}
&C_{\mathbf{M}}  =  \frac{1}{2}\sum_{i=1}^{m}\sum_{j=1}^{N}M_{i,j}(M_{i,j}+1)\nonumber\\
&\phantom{C_{\mathbf{M}}} = \frac{1}{2}\sum_{k=1}^{T}e_{k}\sum_{j=1}^{\ell}m_{j,k}(m_{j,k}+1)\nonumber \\
 & \phantom{C_{\mathbf{M}}}= \frac{1}{2}\mu(\mu+1)\left(N-\sum_{k=2}^{T}e_{k}\right)+\frac{1}{2}\sum_{k=2}^{T}e_{k}\sum_{j=1}^{\ell}m_{j,k}(m_{j,k}+1)\label{eq:cost}\end{align}
where (\ref{eqmult1}) and (\ref{eq:cost}) use the fact that the multiplicity type 1 is always
assumed to be $(\mu,0,\ldots,0)$.

Hence, we obtain
\begin{multline*}
2C_{\mathbf{M}} \! - \! (\mu \! + \! 1)S_{\mathbf{M}}=\mu(\mu+1)\nu_{1,1} + \sum_{k=2}^{T} \! (\mu \! + \! 1)\sum_{j=1}^{\ell}m_{j,k}\nu_{j,k}- \sum_{k=2}^{T} e_{k}\sum_{j=1}^{\ell}m_{j,k}(\mu-m_{j,k}), 
\end{multline*}
and therefore, since $\mu$ and $\nu_{1,1}$ are non-negative, Lemma \ref{lem:GenASD} holds if we can show
\begin{equation}
 (\mu+1)\sum_{j=1}^{\ell}m_{j,k}\nu_{j,k}\geq e_{k}\sum_{j=1}^{\ell}m_{j,k}(\mu-m_{j,k})\label{eqlm1_0}
\end{equation}
for every $k=2,\ldots T$.

Next, we observe that \begin{equation}
(\mu+1)\sum_{j=1}^{\ell}m_{j,k}\nu_{j,k}\geq(\mu+1)\left(\sum_{j:m_{j,k}\neq0}\nu_{j,k}\right)\min_{j:m_{j,k}\neq0}m_{j,k}\label{eqlm1_1}
\end{equation}
and \begin{align}
& \sum_{j:m_{j,k}\neq0}\nu_{j,k} = \sum_{j:m_{j,k}\neq0}(e_{k}-\chi_{j,k})\label{eqlm1_6}\\
& \phantom{...................} = e_{k}|\{j:m_{j,k}\neq0\}|-\sum_{j:m_{j,k}\neq0}\chi_{j,k}\nonumber\\
& \phantom{....................} {\geq} e_{k}(|\{j:m_{j,k}\neq0\}|-1)\label{eqlm1_2}
\end{align}
where (\ref{eqlm1_6}) follows from (\ref{eqlm1_4}) and (\ref{eqlm1_2}) follows from \[
\sum_{j:m_{j,k}\neq0}\chi_{j,k}\leq\sum_{j=0}^{\ell}\chi_{j,k}=e_{k}.\]

From (\ref{eqlm1_1}) and (\ref{eqlm1_2}), we have
\begin{align}
&(\mu+1)\sum_{j=1}^{\ell}m_{j,k}\nu_{j,k}\geq e_{k}(\mu+1)(|\{j:m_{j,k}\neq0\}|-1)\min_{j:m_{j,k}\neq0}m_{j,k}\label{eqlm1_3}\end{align}
and this motivates our definition of allowable multiplicity types.


Specifically, if we choose $\{m_{1,k},m_{2,k},\ldots,m_{\ell,k}\}$
in the allowable set $\mathcal{A}(\mu,\ell)$, defined in (\ref{eq:AllowableMAS}),
then by combining with (\ref{eqlm1_3}), we obtain (\ref{eqlm1_0}) and this completes the proof.\end{proof}
\begin{cor}
\label{cor:ageqmu}With the setting as in Lemma \ref{lem:GenASD},
the integer $a$ in Condition \ref{con:ASDcond} must satisfy $a\geq\mu$.\end{cor}
\begin{proof}
From $(a+1)\left[S_{\mathbf{M}}-\frac{a}{2}(K-1)\right]>C_{\mathbf{M}}$
and $S_{\mathbf{M}}\leq(a+1)(K-1)$ in (\ref{eq:maincond}) and (\ref{eq:piecewise}),
we have  \begin{align*}&(a+1)S_{\mathbf{M}}-C_{\mathbf{M}}>\frac{1}{2}a(a+1)(K-1)\\
&\phantom{(a+1)S_{\mathbf{M}}-C_{\mathbf{M}}}\geq\frac{1}{2}aS_{\mathbf{M}}\end{align*}
and this implies that \begin{equation}2C_{\mathbf{M}}<(a+2)S_{\mathbf{M}}.\label{eqcol1}\end{equation}
But, Lemma \ref{lem:GenASD} states that $2C_{\mathbf{M}}\geq(\mu+1)S_{\mathbf{M}}$.
Combining this with (\ref{eqcol1}) gives a
contradiction unless $a>\mu-1$.
\end{proof}
In Condition \ref{con:ASDcond}, if we carefully design a distortion
measure then for every $a\geq\mu,$ the first constraint (\ref{eq:maincond})
can be equivalently converted to the form where distortion is smaller
than a fixed threshold. 
\begin{thm}
\label{thm:GenASD-allrate} Consider an $(N,K)$ RS code and a MAS($\mu,\ell$)
for {}``top-$\ell$'' decoding with multiplicity matrix $\mathbf{M}$
that only uses $T$ multiplicity types in the allowable set $\mathcal{A}(\mu,\ell)$.
Consider an arbitrary integer $a\geq\mu$. Let $\delta_{a}:\mathcal{X}\times\hat{\mathcal{X}}\rightarrow\mathbb{R}_{\geq0}$,
where in this case $\mathcal{X}=\mathbb{Z}_{\ell+1}$ and $\mathcal{\hat{X}}=\mathbb{Z}_{T+1}\setminus\{0\}$,
be a \emph{letter-by-letter} distortion measure defined by $\delta_{a}(x,\hat{x})=[\Delta_{a}]_{x,\hat{x}}$,
where $\Delta_{a}$ is the $(\ell+1)\times T$ matrix%
\footnote{The first column of $\Delta_{a}$ is $[\frac{2\mu}{a},0,\frac{2\mu}{a},\frac{2\mu}{a},\ldots,\frac{2\mu}{a}]^{T}$
since multiplicity type 1 is always chosen to be $(\mu,0,0,\ldots,0)$. %
}\begin{equation}
\Delta_{a}=\left(\begin{array}{cccc}
\rho_{1,a} & \rho_{2,a} & \ldots & \rho_{T,a}\\
\rho_{1,a}-\frac{2m_{1,1}}{a} & \rho_{2,a}-\frac{2m_{1,2}}{a} & \ldots & \rho_{T,a}-\frac{2m_{1,T}}{a}\\
\rho_{1,a}-\frac{2m_{2,1}}{a} & \rho_{2,a}-\frac{2m_{2,2}}{a} & \ldots & \rho_{T,a}-\frac{2m_{2,T}}{a}\\
\vdots & \vdots & \ddots & \vdots\\
\rho_{1,a}-\frac{2m_{\ell,1}}{a} & \rho_{2,a}-\frac{2m_{\ell,2}}{a} & \ldots & \rho_{T,a}-\frac{2m_{\ell,T}}{a}\end{array}\right)\label{eq:dstmASDmu_allrate}\end{equation}
with \[\rho_{k,a}=\frac{\mu(2a+1-\mu)}{a(a+1)}+\sum_{j=1}^{\ell}\frac{m_{j,k}(m_{j,k}+1)}{a(a+1)}\]
for $k=1,\ldots,T$. Then, the equation (\ref{eq:maincond}) in Condition
\ref{con:ASDcond} is equivalent to \begin{equation*}
d(x^{N},\hat{x}^{N})<\frac{\mu(2a+1-\mu)}{a(a+1)}N-K+1\triangleq D_{a},\label{eq:scasd1_allrate}\end{equation*}
\end{thm}
and it is easy to verify that $D_{\mu}=N-K+1$.
\begin{proof}
First, we show that $\Delta_{a}$ consists of non-zero entries. It
suffices to show that $\rho_{k,a}\geq\frac{2m_{j,k}}{a}$ for all $j=1,\ldots,\ell$ and $k=1,\ldots,T$, i.e.,
\[\mu(2a+1-\mu)+\sum_{j'=1}^{\ell}m_{j',k}(m_{j',k}+1)  \geq  2m_{j,k}(a+1)\]
which is equivalent to
\begin{equation}
2(a+1)(\mu-m_{j,k})+\sum_{j'=1}^{\ell}m_{j',k}(m_{j',k}+1)-\mu(\mu+1)\geq0.\label{eq:nonzeroentries}\end{equation}
This is true since the left hand side of (\ref{eq:nonzeroentries})
is at least 
\begin{align*}
&2(\mu+1)(\mu-m_{j,k})+m_{j,k}(m_{j,k}+1)-\mu(\mu+1)=(\mu-m_{j,k})(\mu+1-m_{j,k})\geq 0.\end{align*}

With the same $e_{k},\nu_{j,k},\chi_{j,k}$ as defined
in the proof of Lemma \ref{lem:GenASD} and the chosen distortion
matrix $\Delta_{a}$, we have
\begin{align*}
&d(x^{N},\hat{x}^{N}) =\sum_{k=1}^{T}\left(\sum_{j=1}^{\ell}\left(\rho_{k,a}-\frac{2m_{j,k}}{a}\right)\chi_{j,k}+\rho_{k,a}\chi_{0,k}\right)\\
& \phantom{d(x^{N},\hat{x}^{N})} =\sum_{k=1}^{T}\left(\rho_{k,a}\sum_{j=0}^{\ell}\chi_{j,k}-2\sum_{j=1}^{\ell}\frac{m_{j,k}}{a}\chi_{j,k}\right)\\
& \phantom{d(x^{N},\hat{x}^{N})} =\sum_{k=1}^{T}\left(\rho_{k,a}e_{k}-2\sum_{j=1}^{\ell}\frac{m_{j,k}}{a}\chi_{j,k}\right).\end{align*}
Noting that the first column of $\Delta_{a}$ is always $[\frac{2\mu}{a},0,\frac{2\mu}{a},\frac{2\mu}{a},\ldots,\frac{2\mu}{a}]^{T}$
and $\nu_{1,1}=e_{1}-\chi_{1,1}$, we obtain\begin{equation}
d(x^{N},\hat{x}^{N})=\frac{2\mu}{a}\nu_{1,1}+\sum_{k=2}^{T}\rho_{k,a}e_{k}-2\sum_{k=2}^{T}\sum_{j=1}^{\ell}\frac{m_{j,k}}{a}\chi_{j,k}.\label{eq:dxx}\end{equation}
Next, one can see that (\ref{eq:maincond}) can be rewritten as\[
\frac{2S_{\mathbf{M}}}{a}-K+1>\frac{2C_{\mathbf{M}}}{a(a+1)}\]
which, by substituting $S_{\mathbf{M}}$ and $C_{\mathbf{M}}$ in (\ref{eq:score})
and (\ref{eq:cost}), is equivalent to 
\begin{align*}
&\frac{2\mu}{a}\!\left(\!N\!-\!\sum_{k=2}^{T}e_{k}\!-\!\nu_{1,1}\!\right)\!\!+\!2\sum_{k=2}^{T}\sum_{j=1}^{\ell}\frac{m_{j,k}}{a}\chi_{j,k}\!-\!K\!+\!1\!>\!\frac{\mu(\mu+1)}{a(a+1)}\!\left(\!N\!-\!\sum_{k=2}^{T}\!e_{k}\!\right)\!\!+\!\!\sum_{k=2}^{T}\!e_{k}\!\sum_{j=1}^{\ell}\!\frac{m_{j,k}(m_{j,k}+1)}{a(a+1)}.
\end{align*}

Equivalently, this gives
\begin{align*}
&\left(\!\frac{2\mu}{a}\!-\!\frac{\mu(\mu+1)}{a(a+1)}\!\right)\!\!N\!-\!K\!+\!1\! >\! \frac{2\mu}{a}\nu_{1,1}\!-\!2\!\sum_{k=2}^{T}\!\sum_{j=1}^{\ell}\!\frac{m_{j,k}}{a}\!\chi_{j,k}\!+\!\!\sum_{k=2}^{T}\!e_{k}\!\!\left(\!\frac{2\mu}{a}\!-\!\frac{\mu(\mu+1)}{a(a+1)}\!+\!\sum_{j=1}^{\ell}\frac{m_{j,k}\!(m_{j,k}+1)}{\mu(\mu+1)}\!\right)
\end{align*}
which in turn is equivalent to 
\begin{align}
&\frac{\mu(2a+1-\mu)}{a(a+1)}N-K+1  > \frac{2\mu}{a}\nu_{1,1}+\sum_{k=2}^{T}e_{k}\rho_{k,a}-\frac{2}{a}\sum_{k=2}^{T}\sum_{j=1}^{\ell}m_{j,k}\chi_{j,k}.\label{eq:dxx2}\end{align}


Finally, combining (\ref{eq:dxx}) and (\ref{eq:dxx2}) gives the proof.\end{proof}
\begin{example}
Consider mASD-2 for $a=\mu=2$. In this case, the distortion matrix is\begin{equation}
\Delta=\left(\begin{array}{cccc}
2 & \nicefrac{5}{3} & 2 & 1\\
0 & \nicefrac{2}{3} & 2 & 1\\
2 & \nicefrac{2}{3} & 0 & 1\end{array}\right).\label{eq:dstmASD2}\end{equation}

\end{example}
However, Condition \ref{con:ASDcond} also requires the second constraint
(\ref{eq:piecewise}) to be satisfied. In addition, we need to choose
an integer $a\geq\mu$ in order to apply our proposed approach. Therefore,
we first consider the case of high-rate RS codes where if $a=\mu$
then the satisfaction of (\ref{eq:maincond}) also implies the satisfaction
of (\ref{eq:piecewise}). For the case of lower-rate RS codes, we
obtain a range of $a$ and also propose a heuristic method to choose
an appropriate $a$.

\subsubsection{High-rate Reed-Solomon codes}

In this subsection, we focus on high-rate RS codes which are usually
seen in many practical applications. The high-rate constraint allows
us to see that $a=\mu$ is essentially the correct choice.
\begin{lemma}
\label{lem:Midhighrate}Consider an $(N,K)$ RS code with rate \[\frac{K}{N}\geq\frac{1}{N}+\frac{\mu}{\mu+1}.\]
If equation (\ref{eq:maincond}) is satisfied for $a=\mu$, or equivalently,
\[d(x^{N},\hat{x}^{N})<N-K+1\] under the distortion measure $\Delta_{\mu}$
then whole Condition \ref{con:ASDcond} is satisfied and the transmitted
codeword will be therefore on the list.\end{lemma}
\begin{proof}
Suppose (\ref{eq:maincond}) is satisfied for $a=\mu$, i.e., 
\begin{equation}S_{\mathbf{M}}>\frac{C_{\mathbf{M}}}{\mu+1}+\frac{\mu}{2}(K-1).\label{eqlm2_0}\end{equation}
                                                                                              
We will show that 
\begin{align}
 &\mu(K-1)<S_{\mathbf{M}}\label{eqlm2_1}\\
 &\phantom{...............}\leq(\mu+1)(K-1)\label{eqlm2_2}
\end{align}
and, therefore, both (\ref{eq:maincond}) and (\ref{eq:piecewise})
in Condition \ref{con:ASDcond} are satisfied for $a=\mu$.

Firstly, using Lemma \ref{lem:GenASD} we have \[
                                               \frac{S_{\mathbf{M}}}{2}\geq S_{\mathbf{M}}-\frac{C_{\mathbf{M}}}{\mu+1}
                                              \]
and consequently, (\ref{eqlm2_1}) is implied
by (\ref{eqlm2_0})
since\[
\frac{S_{\mathbf{M}}}{2}\geq S_{\mathbf{M}}-\frac{C_{\mathbf{M}}}{\mu+1}>\frac{\mu}{2}(K-1).\]
Secondly, note that (\ref{eqlm2_2}) holds since
\begin{align}
&S_{\mathbf{M}} =\mu\left(N-\sum_{k=2}^{T}e_{k}-\nu_{1,1}\right)+\sum_{k=2}^{T}\sum_{j=1}^{\ell}m_{j,k}(e_{k}-\nu_{j,k})\nonumber\\
& \phantom{\mathbb{E}} =\mu N-\mu\nu_{1,1}-\sum_{k=2}^{T}\sum_{j=0}^{\ell}m_{j,k}\nu_{j,k}-\sum_{k=2}^{T}e_{k}\left(\mu-\sum_{j=1}^{\ell}m_{j,k}\right)\nonumber\\
& \phantom{\mathbb{E}} \leq \mu N\label{eqlm2_3}\\
& \phantom{\mathbb{E}} \leq (\mu+1)(K-1)\label{eqlm2_4}\end{align}
where (\ref{eqlm2_3}) is obtained by dropping non-negative terms and (\ref{eqlm2_4}) follows from the high-rate constraint $\frac{K-1}{N}\geq\frac{\mu}{\mu+1}$.

Finally, by Theorem \ref{thm:GenASD-allrate}, one can verify that
equation (\ref{eq:maincond}) with $a=\mu$ is equivalent to \[d(x^{N},\hat{x}^{N})<D_{\mu}=N-K+1\]
under the distortion measure $\Delta_{\mu}$. 
\end{proof}
However, there are possibly other integers $a\neq\mu$ that can also
satisfy Condition \ref{con:ASDcond}. If we consider higher-rate RS
codes, as in the following theorem, then we can claim that $a=\mu$ is the
only such integer.
\begin{thm}
\label{thm:GenASD} Consider an $(N,K)$ RS code with rate \[\frac{K}{N}\geq\frac{1}{N}+\frac{\mu(\mu+3)}{(\mu+1)(\mu+2)}.\]
The integer $a$ in Condition \ref{con:ASDcond} must satisfy $a=\mu$
and, consequently, the set of constraints (\ref{eq:maincond}) and (\ref{eq:piecewise})
in Condition \ref{con:ASDcond} is equivalent to\[
d(x^{N},\hat{x}^{N})<N-K+1\]
 under the distortion measure $\Delta_{\mu}$. \end{thm}
\begin{proof}
We first see that \[(a+1)\left[S_{\mathbf{M}}-\frac{a}{2}(K-1)\right]>C_{\mathbf{M}}\]
in (\ref{eq:maincond}) implies \[S_{\mathbf{M}}-\frac{a}{2}(K-1)>\frac{C_{\mathbf{M}}}{a+1}\]
and, with the score $S_{\mathbf{M}}$ and the cost $C_{\mathbf{M}}$
computed in (\ref{eq:score}) and (\ref{eq:cost}), we obtain

\begin{align*}
&\mu\!\left(\!N\!-\!\sum_{k=2}^{T}e_{k}\!-\!\nu_{1,1}\!\right)\!+\!\sum_{k=2}^{T}\sum_{j=1}^{\ell}m_{j,k}(e_{k}\!-\!\nu_{j,k})\!-\!\frac{a}{2}(K-1) \nonumber \\
& \phantom{\mathbb{EEEEEEEEEEEEEEEEEEEEEEEEEEEE}} >\frac{\mu(\mu+1)}{2(a+1)}\left(N-\sum_{k=2}^{T}e_{k}\right)+\sum_{k=2}^{T}e_{k}\sum_{j=1}^{\ell}\frac{m_{j,k}(m_{j,k}+1)}{2(a+1)}.
\end{align*}

This gives
\begin{align}
 & \left(\!\mu\!-\!\frac{\mu(\mu+1)}{2(a+1)}\!\right)\!N\!-\!\frac{a}{2}(K-1)>\mu\nu_{1,1}\!+\!\sum_{j=2}^{T}\sum_{j=1}^{\ell}\nu_{j,k}\!+\!\sum_{k=2}^{T}\!e_{k}\!\left(\!\mu\!-\!\sum_{j=1}^{\ell}m_{j,k}\!+\!\sum_{j=1}^{\ell}\frac{m_{j,k}(m_{j,k}+1)}{2(a+1)}\right)\nonumber\\
 & \phantom{\left(\!\mu\!-\!\frac{\mu(\mu+1)}{2(a+1)}\!\right)\!N\!-\!\frac{a}{2}(K-1)\!} \geq\sum_{k=2}^{T}e_{k}\left(\mu-\sum_{j=1}^{\ell}m_{j,k}\right)\label{eqthm2_5}\\
 & \phantom{\left(\!\mu\!-\!\frac{\mu(\mu+1)}{2(a+1)}\!\right)\!N\!-\!\frac{a}{2}(K-1)\!} \geq 0\label{eq:aconstr}
\end{align}
where (\ref{eqthm2_5}) is obtained by dropping non-negative terms.

Combining this inequality with the high-rate constraint implies that
\[
\frac{\mu(2a+1-\mu)}{a(a+1)}>\frac{K-1}{N}\geq\frac{\mu(\mu+3)}{(\mu+1)(\mu+2)}\]
which leads to $a<\mu+1$, i.e. $a\leq\mu$. 

This, together with $a\geq\mu$ according to Corollary \ref{cor:ageqmu},
leave $a=\mu$ as the only possible choice. Finally, by seeing that
\begin{align*}
& \frac{K}{N}\geq\frac{1}{N}+\frac{\mu(\mu+3)}{(\mu+1)(\mu+2)}>\frac{1}{N}+\frac{\mu}{\mu+1}
\end{align*}
and applying Lemma \ref{lem:Midhighrate} we conclude the proof.\end{proof}
\begin{cor}
\label{cor:For-mASD}When the RD approach is used, $R(D)$ is positive
for $D_{\min}\leq D<D_{\max}$ and is zero for $D\geq D_{\max}$.
Computing $D_{\max}$ reveals how good the distortion measure matrix
is at rates close to zero (i.e., the erasure codebook has only one
entry). For mASD-$\mu$, \begin{align*}
\begin{split}
&D_{\max}(\mbox{mASD-}\mu)=\sum_{i=1}^{N}\min_{k=2,\ldots,T}\left\{ 2(1-p_{i,1}),\rho_{k,\mu}-\sum_{j=1}^{\ell}\frac{m_{j,k}}{\mu}p_{i,j}\right\} 
\end{split}
\end{align*}
while for mBM-$\ell$, \begin{align*}
D_{\max}(\mbox{mBM-}\ell)=\sum_{i=1}^{N}\min\{1,2(1-p_{i,1})\}.\end{align*}
Moreover, if mASD-$\mu$ uses multiplicity type $(0,0,\ldots0$) then
$D_{\max}($mASD-$\mu)\leq D_{\max}($mBM-$\ell$) for every $\mu,\ell$.\end{cor}
\begin{proof}
See Appendix \ref{sec:AppCorDmax}.\end{proof}
\begin{example}
Consider mASD-2 with distortion matrix in (\ref{eq:dstmASD2}). We
have
\begin{align*}
&D_{\max}(\mbox{mASD-}2)=\sum_{i=1}^{N}\min\left\{1,2(1-p_{i,1}),\frac{5}{3}-\frac{2}{3}(p_{i,1}+p_{i,2})\right\} 
\end{align*} 
 which is less than or equal to $D_{\max}(\mbox{mBM-}\ell)$ for every
$\ell$. This fact can be seen in Fig. \ref{fig:rdcurve} which is
obtained by simulation. This also predicts that, as expected, ASD
decoding will be superior when $R$ is small.
\end{example}

\subsubsection{Lower-rate Reed-Solomon codes}

Without the high-rate constraint as in Theorem \ref{thm:GenASD},
we may not have $a=\mu$. However, we can obtain a range for $a$
and heuristically choose the integer $a$ that potentially give the
highest rate-distortion exponent. After that, we can also apply the
algorithms proposed in Section \ref{sec:Proposed-Algorithm} with
the corresponding distortion measure $\Delta_{a}$ and distortion
threshold $D_{a}$ derived in Theorem \ref{thm:GenASD-allrate}.
\begin{table}[t!]
\caption{Example ranges of possible $a$}
\label{table_1}
\centering{}\begin{tabular}{|c|c|c|}
\hline 
 & RS(255,191) & RS(255,127)\tabularnewline
\hline
\hline 
$\mu=2$ & $2\leq a\leq3$ & $2\leq a\leq6$\tabularnewline
\hline 
$\mu=3$ & $3\leq a\leq5$ & $3\leq a\leq9$\tabularnewline
\hline
\end{tabular}
\end{table}

The following lemma tells us the range of possible $a$.
\begin{lemma}
\label{lem:rangea}Consider an ($N$,$K$) RS code. In order to satisfy
(\ref{eq:maincond}), one must have\[ \mu\leq a\leq\Big\lceil\mu\theta-\nicefrac{1}{2}+\sqrt{\mu^{2}\theta\left(\theta-1\right)+\nicefrac{1}{4}}\Big\rceil-1\]%
where $\theta\triangleq\frac{N}{K-1}$.\end{lemma}
\begin{proof}
First note that (\ref{eq:aconstr}) holds for any $(N,K)$. Therefore,
we have \[
\mu-\frac{\mu(\mu+1)}{2(a+1)}>\frac{a(K-1)}{2N}.\]
Combining this with $a\geq\mu$ in Corollary \ref{cor:ageqmu}, we obtain the stated result.
\end{proof}

\begin{example}
 Table \ref{table_1} gives several example ranges of possible $a$ for some choices of $\mu$ and RS codes.
\end{example}

Among possible choices of $a$, we are interested in choosing $a$
that gives the largest rate-distortion exponent and therefore has
a better chance to satisfy Condition \ref{con:ASDcond}. The following
lemma can give us an insight of how to choose such an integer $a$.
\begin{lemma}
If\begin{equation}
a>\frac{1}{2}\left(\sqrt{1+4\theta\mu(\mu+1)}-3\right)\label{eq:aFdec}\end{equation}
 where $\theta=\frac{N}{K-1}$ then starting from $a$, the rate-distortion
exponent $F_{a}$ strictly decreases until reaching zero, i.e., $F_{a}>F_{a+1}>F_{a+2}>\ldots\geq0$
if rate $R$ is fixed.\end{lemma}
\begin{proof}
For a fixed rate $R$, the distortion measure $\Delta_{a+1}$ and
distortion $D_{a+1}$ yield exponent $F_{a+1}$. Scaling both $\Delta_{a+1}$
and $D_{a+1}$ leaves $F_{a+1}$ unchanged. Hence, $\frac{a+1}{a}\Delta_{a+1}$
and $\frac{a+1}{a}D_{a+1}$ also yield $F_{a+1}$. Next, we will show
that\begin{equation}
\frac{a+1}{a}\Delta_{a+1}\geq\Delta_{a}.\label{eq:lowerdist}\end{equation}

To prove (\ref{eq:lowerdist}), it suffices to show
\begin{equation}
\frac{a+1}{a}\rho_{k,a+1}\geq\rho_{k,a}\label{rhocond}
\end{equation}
since \[\frac{a+1}{a}\left(\rho_{k,a+1}-\frac{2m_{j,k}}{a+1}\right)\geq\rho_{k,a}-\frac{2m_{j,k}}{a}\]
is also equivalent to (\ref{rhocond}).

Equivalently, we need to show \[
\mu(\mu+1)\geq\sum_{j=1}^{\ell}m_{j,k}(m_{j,k}+1)\]
which is true because $\mu\geq\sum_{i=1}^{\ell}m_{j,k}$ by the definition
of allowable multiplicity types.

Thus, (\ref{eq:lowerdist}) holds and, therefore, the exponent yielded
by $\Delta_{a}$ and $\frac{a+1}{a}D_{a+1}$ is at least $F_{a+1}.$
From (\ref{eq:aFdec}) we have \begin{align*}
&D_{a}=\frac{\mu(2a+1-\mu)}{a(a+1)}N-K+1\\
&\phantom{D_{a}}>\frac{\mu(2a+3-\mu)}{a(a+2)}N-\frac{a+1}{a}(K-1)\\
&\phantom{D_{a}}=\frac{a+1}{a}D_{a+1}.\end{align*}

Since for a fixed $R$, exponent $F$ is increasing in distortion
$D$ \cite[Thm 6.6.2]{Blahut-1987}, we know that $F_{a}>F_{a+1}$
where $F_{a}$ is the exponent yielded by $\Delta_{a}$ and $D_{a}$.\end{proof}
\begin{table}
\caption{Example ranges of $a$ that gives the largest exponent}
\label{table_2}
\centering{}\begin{tabular}{|c|c|c|}
\hline 
 & RS(255,191) & RS(255,127)\tabularnewline
\hline
\hline 
$\mu=2$ & $a=2$ & $a\in\{2,3\}$\tabularnewline
\hline 
$\mu=3$ & $a=3$ & $a\in\{3,4\}$\tabularnewline
\hline 
$\mu=12$ & $a\in\{12,13\}$ & $12\leq a\leq17$\tabularnewline
\hline
\end{tabular}
\end{table}

\begin{figure}[t]
\centering{}\includegraphics[scale=0.85]{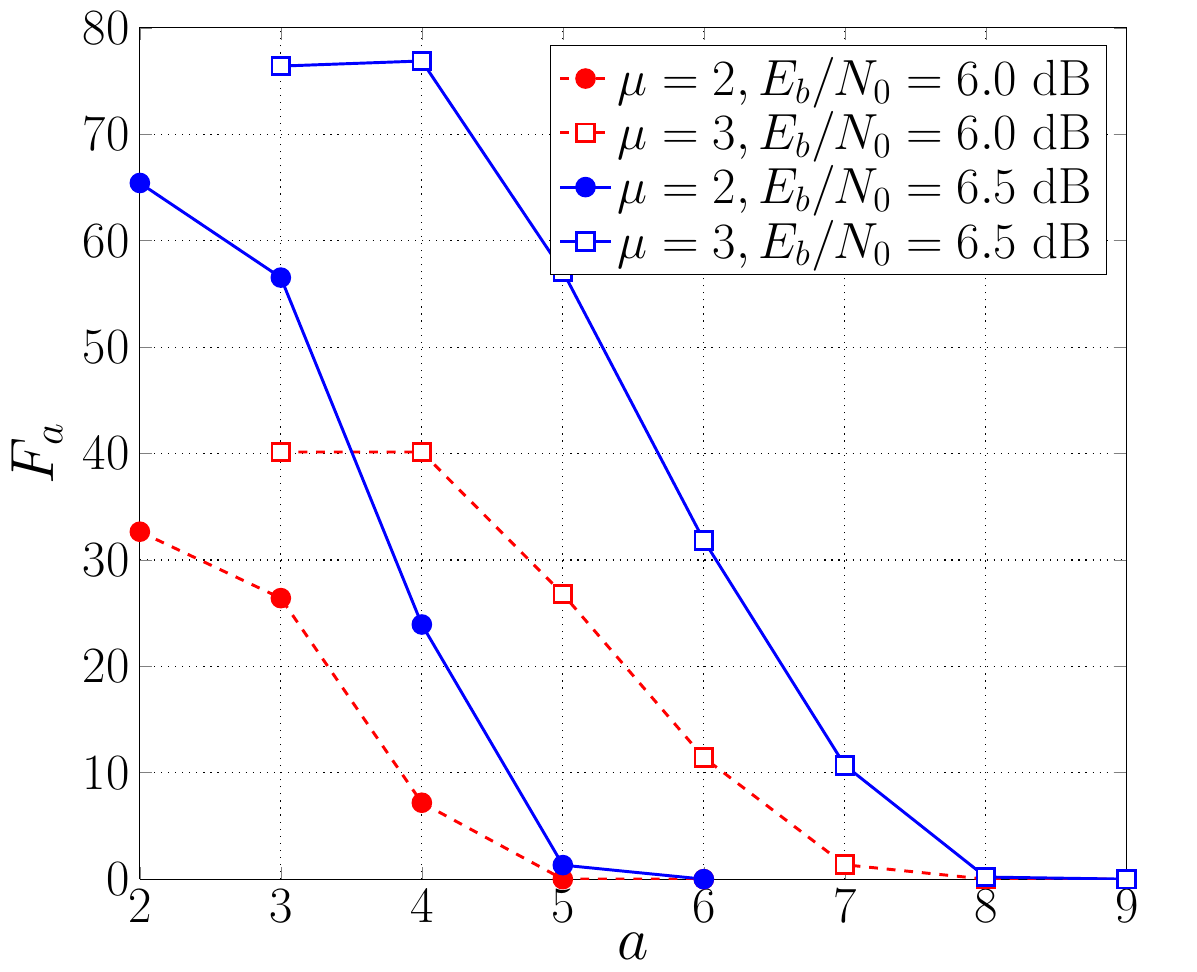}\caption{\label{fig:Fvsa} Plot of exponent $F_{a}$ versus $a$ for $\mu=2$
and $\mu=3$ with a fixed rate $R=6$. Simulations are conducted for
the (255,127) RS code using BPSK over an AWGN channel at $E_{b}/N_{0}=6.0$
dB and $6.5$ dB.}

\end{figure}[h]
\begin{cor}
The integer $a$ that gives the largest exponent lies in the range\[ \mu\leq a\leq\Big\lfloor\frac{1}{2}\left(\sqrt{1+4\theta\mu(\mu+1)}-3\right)\Big\rfloor+1.\] 
\end{cor}

\begin{example}
 The following Table \ref{table_2} presents several example ranges of $a$ that gives the largest exponent for some choices of $\mu$ and RS codes.
\end{example}

\begin{remrk}
Simulation results also confirm our analysis. For example, in Fig.
\ref{fig:Fvsa}, $a=3$ and $a=4$ give roughly same and the largest
exponents for $\mu=3$ while $a=2$ yields the largest exponent for
$\mu=2$. In fact, simulation results suggest that, typically, either
$a=\mu$ or $a=\mu+1$ gives the best exponent.

In Condition \ref{con:ASDcond}, for lower-rate RS codes, so far we
have only paid attention to (\ref{eq:maincond}). However, it is also
required that \[a(K-1)<S_{\mathbf{M}}\leq(a+1)(K-1),\] or equivalently
\begin{equation}
a+1=\Big\lceil\frac{S_{\mathbf{M}}}{K-1}\Big\rceil. \label{acondt}
\end{equation}
 While it is hard to tell exactly which $a$ will satisfy (\ref{acondt})
with high probability right away, we can propose a heuristic method
to choose the integer $a$ that is likely to work. We first need the
following lemma.\end{remrk}
\begin{lemma}
Suppose we have obtained a test-channel input-probability distribution
matrix $\mathbf{Q}$ (e.g., during Step 2a or Step 2b in the proposed
algorithms in Section \ref{sec:Proposed-Algorithm}) and the set of
erasure patterns for mASD is generated independently and randomly
according to $\mathbf{Q}$. Then, the expected score can be computed
as follows:\begin{align}
\mathbb{E}[S_{\mathbf{M}}]=\sum_{k=1}^{T}\sum_{j=1}^{\ell}\sum_{i=1}^{N}m_{j,k}p_{i,j}q_{i,k}.\label{eq:Escore}\end{align}
\end{lemma}
\begin{proof}
The proof follows from the following equations: \begin{align}
&\mathbb{E}[S_{\mathbf{M}}]= \mathbb{E}\left[\sum_{k=1}^{T}\sum_{j=1}^{\ell}m_{j,k}\chi_{j,k}\right]\label{eqES}\\ 
 &\phantom{\mathbb{E}[S_{\mathbf{M}}]} =  \sum_{k=1}^{T}\sum_{j=1}^{\ell}m_{j,k}\mathbb{E}[\chi_{j,k}]\nonumber\\
 &\phantom{\mathbb{E}[S_{\mathbf{M}}]} =  \sum_{k=1}^{T}\sum_{j=1}^{\ell}m_{j,k}\mathbb{E}\left[\sum_{i=1}^{N}\mathbbm{1}_{\{X_{i}=j,\hat{X}_{i}=k\}}\right]\nonumber\\
 &\phantom{\mathbb{E}[S_{\mathbf{M}}]} =  \sum_{k=1}^{T}\sum_{j=1}^{\ell}\sum_{i=1}^{N}m_{j,k}\Pr(X_{i}=j,\hat{X}_{i}=k)\nonumber\\
 &\phantom{\mathbb{E}[S_{\mathbf{M}}]} =  \sum_{k=1}^{T}\sum_{j=1}^{\ell}\sum_{i=1}^{N}m_{j,k}p_{i,j}q_{i,k}\nonumber\end{align}
where (\ref{eqES}) is implied by (\ref{eqscorechi}).
\end{proof}
Next, we propose a heuristic method to find the appropriate integer
$a$ to work with as follows.

\begin{alg}~
\begin{itemize}
\item Step 1: Start with $a=\mu$, using distortion measure $\Delta_{a}$
and distortion threshold $D_{a}$ to get the corresponding distribution
matrix $\mathbf{Q}$ as discussed above.
\item Step 2: Compute the expected score $\mathbb{E}[S_{\mathbf{M}}]$ using
(\ref{eq:Escore}). If $\Big\lceil\frac{\mathbb{E}[S_{\mathbf{M}}]}{K-1}\Big\rceil=a+1$
then output $a$ and stop. If not set $a\leftarrow a+1$ and return
to Step 1.
\end{itemize}
\end{alg}
\begin{remrk}
In simulations with small to moderate $\mu$, it is usually found that
$a$ is either $\mu$ or $\mu+1$. Typically, $\frac{\mathbb{E}[S_{\mathbf{M}}]}{K-1}>\mu$
and a unit increase of $a$ produces a small increase in $\frac{\mathbb{E}[S_{\mathbf{M}}]}{K-1}$.
\end{remrk}

\begin{remrk}
So far, we have considered only the allowable multiplicity types in
Definition \ref{def:AllowableTypes}. It is possible to obtain better
performance if we relax some constraints and allow multiplicity types
to be in the relaxed set \[
\mathcal{A}_{0}(\mu,\ell)\triangleq\left\{ (m_{1},m_{2},\ldots,m_{\ell})\Big|\begin{array}{c}
\sum_{j=1}^{\ell}m_{j}\leq\mu\end{array}\right\} .\]
In this case, some theoretical results, e.g., results in Lemma 1 and
Theorem 2, do not hold. However, this modification combined with the
heuristic method above can improve the decoding performance, especially
with large $\mu$. Specifically, we consider mASD$_{0}$-$\mu$ which
denotes our proposed multiple ASD decoding algorithm that only uses
multiplicity types $(0,0)$ and$(m_{1},m_{2})$ of the form $m_{1}+m_{2}=\mu$.
These multiplicity types form a subset of $\mathcal{A}_{0}(\mu,2)$.
The choice of $\ell=2$ is suggested by observations that top-$2$
decoding performs almost as good as top-$\ell$ decoding for $\ell>2$.
The integer $a$ used in mASD$_{0}$-$\mu$ is found through the heuristic
method. In Fig. \ref{fig:Fvsa-410}, simulations are conducted for
the (458,410) RS code using BPSK over an AWGN channel. For $\mu=10,$
it can again be observed that $a=\mu$ gives the best exponent. More
simulation results of this heuristic method can be seen in Section
\ref{sec:Simulation-results}.
\end{remrk}
\begin{figure}[t]
\centering{}\includegraphics[scale=0.85]{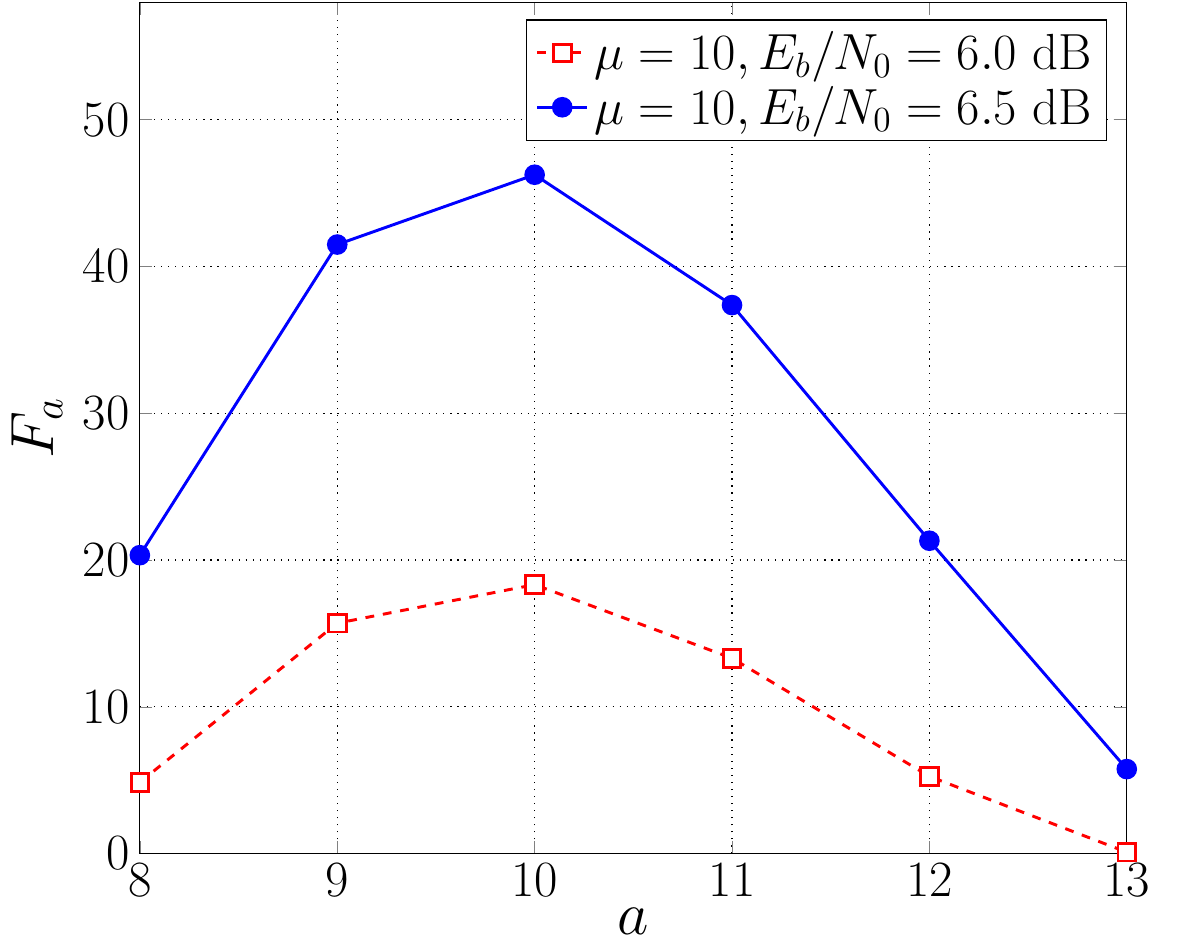}\caption{\label{fig:Fvsa-410} Plot of exponent $F_{a}$ versus $a$ for $\mu=10$
with a fixed rate $R=6$. The set of multiplicity types considered
is the relaxed set $\mathcal{A}_{0}(10,2)$. Simulations are conducted
for the (458,410) RS code over $\mathbb{F}_{2^{10}}$ using BPSK over
an AWGN channel at $E_{b}/N_{0}=6.0$ dB and $6.5$ dB. }

\end{figure}

\section{Closed-Form Analysis of RD and RDE Functions for Some Distortion
Measures\label{sec:Closed-Form-Analysis-of}}

\subsection{Closed-form RD function\label{sub:AnalyticalRD}}

For some simple distortion measures, we can compute the RD functions
analytically in closed form. First, we observe an error pattern as
a sequence of i.n.d. random source components. Then, we compute the
component RD functions at each index of the sequence and use convex
optimization techniques to allocate the total rate and distortion
to various components. This method converges to the solution faster
than the numerical method in Section \ref{sec:Computing-RD}. The
following two theorems describe how to compute the RD functions for
the simple distortion measures of Proposition \ref{prop:bma1} and
\ref{prop:BitASD}.
\begin{lemma}
\label{lem:OneVar}Consider a binary source $X$ where $\Pr(X=1)=p$
and $\Pr(X=0)=1-p$ . With the distortion measure
in (\ref{eq:dstfnBMA}), the rate-distortion function for this source
is%
\footnote{The binary entropy function is $H(u)\triangleq-u\log u-(1-u)\log(1-u)$.%
} \[R(D)=\left[H(p)-H(D+p-1)\right]^{+}.\]\end{lemma}
\begin{proof}
See Appendix \ref{sec:AppLemRDmBM1}.\end{proof}
\begin{thm}
\label{thm:(BMA-RD)}(Conventional errors-and-erasures {}``mBM-1''
decoding) Let $p_{i,1}\triangleq\Pr(X_{i}=1)$ for $i=1,\ldots,N$.
The overall rate-distortion function is given by \[R(D)=\sum_{i=1}^{N}\left[H(p_{i,1})-H(\tilde{D}_{i})\right]^{+}\]
where $\tilde{D}_{i}\triangleq D_{i}+p_{i,1}-1$ and $\tilde{D}_{i}$
can be found be a reverse water-filling procedure (see \cite[Theorem 13.3.3]{Cover-1991}):\[
\tilde{D}_{i}=\begin{cases}
\lambda & \mbox{if}\,\,\lambda<\min\{p_{i,1},1-p_{i,1}\}\\
\min\{p_{i,1},1-p_{i,1}\} & \mbox{otherwise}\end{cases}\]
where $\lambda$ should be chosen so that \[\sum_{i=1}^{N}\tilde{D}_{i}=D+\sum_{i=1}^{N}p_{i,1}-N.\]
The $R(D)$ function can be achieved by the test-channel input-probability
distribution\[
q_{i,0}\triangleq\Pr(\hat{X}_{i}=0)=\frac{1-p_{i,1}-\tilde{D}_{i}}{1-2\tilde{D}_{i}}\]
and
\[q_{i,1}\triangleq\Pr(\hat{X}_{i}=1)=\frac{p_{i,1}-\tilde{D}_{i}}{1-2\tilde{D}_{i}}.\]
\end{thm}
\begin{proof}
See Appendix \ref{sec:ProofThmBMARD}.\end{proof}
\begin{thm}
\label{thm:(bASD-RD)}(Bit-level ASD {}``m-bASD'' decoding) Let
$r_{i,1}\triangleq\Pr(B_{i}=1)$ and $r_{i,0}\triangleq\Pr(B_{i}=0)$
for $i=1,\ldots,n$. The overall rate-distortion function in m-bASD
scheme is given by \[R(D)=\sum_{i=1}^{n}\left[R_{i}(\lambda)\right]^{+}\]
where
\begin{align*}
&R_{i}(\lambda)=H(r_{i,1})-H\left(\frac{1+\lambda}{1+\lambda+\lambda^{2}}\right)+\left(r_{i,1}-\frac{1+\lambda}{1+\lambda+\lambda^{2}}\right)H\left(\frac{\lambda}{1+\lambda}\right)
 \end{align*}
and the distortion component $D_{i}$ is given by\[
D_{i}=\begin{cases}
\frac{1+2\lambda+3\lambda^{2}}{1+\lambda+\lambda^{2}}-r_{i,1}\frac{1+2\lambda}{1+\lambda} & \mbox{if}\,\, R_{i}(\lambda)>0\\
\min\{1,3(1-r_{i,1})\} & \mbox{otherwise}\end{cases}\]
where $\lambda\in(0,1)$ should be chosen so that $\sum_{i=1}^{n}D_{i}=D$.
The $R(D)$ function can be achieved by the following test-channel
input-probability distribution\[
s_{i,0}\triangleq\Pr(\hat{B}_{i}=0)=\frac{(1+\lambda)-r_{i,1}(1+\lambda+\lambda^{2})}{1-\lambda^{2}}\]
and \[ s_{i,1}\triangleq\Pr(\hat{B}_{i}=1)=\frac{r_{i,1}(1+\lambda+\lambda^{2})-\lambda(1+\lambda)}{1-\lambda^{2}}.\]
\end{thm}
\begin{proof}
[Sketch of proof] With the distortion measure in (\ref{eq:dstfnBGMD}),
using the method in \cite[Chapter 2]{Berger-1971} we can compute
the rate-distortion function components 
\begin{align*}&R_{i}(\lambda_{i})=H(r_{i,1})-H\left(\frac{1+\lambda_{i}}{1+\lambda_{i}+\lambda_{i}^{2}}\right)+\left(r_{i,1}-\frac{1+\lambda_{i}}{1+\lambda_{i}+\lambda_{i}^{2}}\right)H\left(\frac{\lambda_{i}}{1+\lambda_{i}}\right)
\end{align*}
where $\lambda_{i}$ is a Lagrange multiplier such that 
\[D_{i}=\frac{1+2\lambda_{i}+3\lambda_{i}^{2}}{1+\lambda_{i}+\lambda_{i}^{2}}-r_{i,1}\frac{1+2\lambda_{i}}{1+\lambda_{i}}\]
for each bit index $i$. Then, the Kuhn-Tucker conditions define the
overall rate allocation using the similar argument as in the proof
of Theorem \ref{thm:(BMA-RD)}.
\end{proof}

\subsection{Closed-form RDE function}

In this subsection, we consider the case mBM-1 whose distortion measure
is given in (\ref{eq:dstfnBMA}). We study the setup that RS codewords
defined over Galois field $\mathbb{F}_{m}$ are transmitted over the
$m$-ary symmetric channel ($m$-SC) which for each parameter $p$
can be modeled as\[
\Pr(r|c)=\begin{cases}
p & \mbox{if}\,\, r=c\\
(1-p)/(m-1) & \mbox{if}\,\, r\neq c\end{cases}.\]
Here, $c$ (resp. $r$) is the transmitted (resp. received) symbol
and $r,c\in\mathbb{F}_{m}$. For this channel model, we restrict our
attention to the range of $p$ where the received symbol is the most-likely
(i.e., $p>(1-p)/(m-1)$). Therefore, at each index $i$ of the codeword,
the hard-decision is also the received symbol and then it is correct
with probability $p$. Thus, we have $p_{i,1}=\Pr(X_{i}=1)=p$ for
every index $i$ of the error pattern $x^{N}$. That means, in this
context we have a source $x^{N}$ with i.i.d. binary components $x_{i}$.
Since the components $x_{i}$'s are i.i.d, we can treat each $x_{i}$
as a binary source $X$ with $\Pr(X=1)=p$ and first compute the RDE
function for this source $X$ as given by an analysis in Appendix
\ref{sec:RDEanalysis}. Based on this analysis, we obtain the following
lemmas and theorems for the mBM-1 decoding algorithm of RS codes over
an $m$-SC channel.
\begin{lemma}
\label{lem:Lemmah}Let $h(u)=H(u)-H(u+D-1)$ map $u\in\left[1-D,1-\frac{D}{2}\right)$
to $R$. Then, the inverse mapping of $h$, \[
h^{-1}:(0,H(1-D)]\to\left[1-D,1-\frac{D}{2}\right),\]
is well-defined and maps $R$ to $u$.\end{lemma}
\begin{proof}$h(u)$ is strictly decreasing since the derivative
is negative over $\left[1-D,1-\frac{D}{2}\right)$. Hence, the mapping
$h:\left[1-D,1-\frac{D}{2}\right)\rightarrow(0,H(1-D)]$ is one-to-one.
From the analysis in Appendix \ref{sec:RDEanalysis}, one can also
see that $h$ is onto.\end{proof}
\begin{thm}
Using mBM-1 with $2^{R}$ decoding attempts where $R\in(0,NH(1-\frac{D}{N})]$,
the maximum rate-distortion exponent that can be achieved is%
\footnote{The Kullback-Leibler divergence is $D_{KL}(u||p)\triangleq u\log\frac{u}{p}+(1-u)\log\frac{1-u}{1-p}$.%
}\begin{equation}
F=N\, D_{KL}\left(h^{-1}\left(\frac{R}{N}\right)\,\bigg|\bigg|\, p\right).\label{eq:ComputeF}\end{equation}
\end{thm}
\begin{proof}
First, note that in our context where we have a source sequence $x^{N}$
of $N$ i.i.d. source components, the rate and exponent for each source
component are now $\frac{R}{N}$ and $\frac{F}{N}$. From Case 3 in
Appendix \ref{sec:RDEanalysis} and from Lemma \ref{lem:Lemmah},
we have \[
\frac{F}{N}=D_{KL}(u||p)=D_{KL}\left(h^{-1}\left(\frac{R}{N}\right)\,\bigg|\bigg|\, p\right)\]
and the theorem follows.\end{proof}
\begin{lemma}
Let $g(u)=D_{KL}(u||p)$ map $u\in[1-D,p]$ to $F$. Then, the inverse
mapping of $g$, \[
g^{-1}:[0,D_{KL}(1-D\,||\, p)]\rightarrow[1-D,p]\]
is well-defined and maps $F$ to $u$.\end{lemma}
\begin{proof}
We first see that $g(u)$ is a strictly convex function and achieves
minimum value at $u=p$ and therefore $g(u)$ is strictly decreasing
over $[1-D,p]$. Thus, the mapping $g:[1-D,p]\to[0,D_{KL}(1-D\,||\, p)]$
is one-to-one. From the analysis in Appendix \ref{sec:RDEanalysis},
one can also see that $g$ is onto.\end{proof}
\begin{thm}
In order to achieve a rate-distortion exponent of $F\in\left[0,N\, D_{KL}\left(1-D\,||\, p\right)\right]$,
the minimum number of decoding attempts required for mBM-1 is $2^{R}$
where\[
R=N\left[H\left(g^{-1}\left(\frac{F}{N}\right)\right)-H\left(g^{-1}\left(\frac{F}{N}\right)+\frac{D}{N}-1\right)\right]^{+}.\]
\end{thm}
\begin{proof}
We also note that the rate, distortion and exponent for each source
component are $\frac{R}{N},\frac{D}{N}$ and $\frac{F}{N}$ respectively.
Combining all the cases in Appendix \ref{sec:RDEanalysis}, we have
\[
\frac{R}{N}=\left[H\left(g^{-1}\left(\frac{F}{N}\right)\right)-H\left(g^{-1}\left(\frac{F}{N}\right)+\frac{D}{N}-1\right)\right]^{+}\]
and the theorem follows. 
\end{proof}
\begin{figure}[t]
\centering{}\includegraphics[scale=0.85]{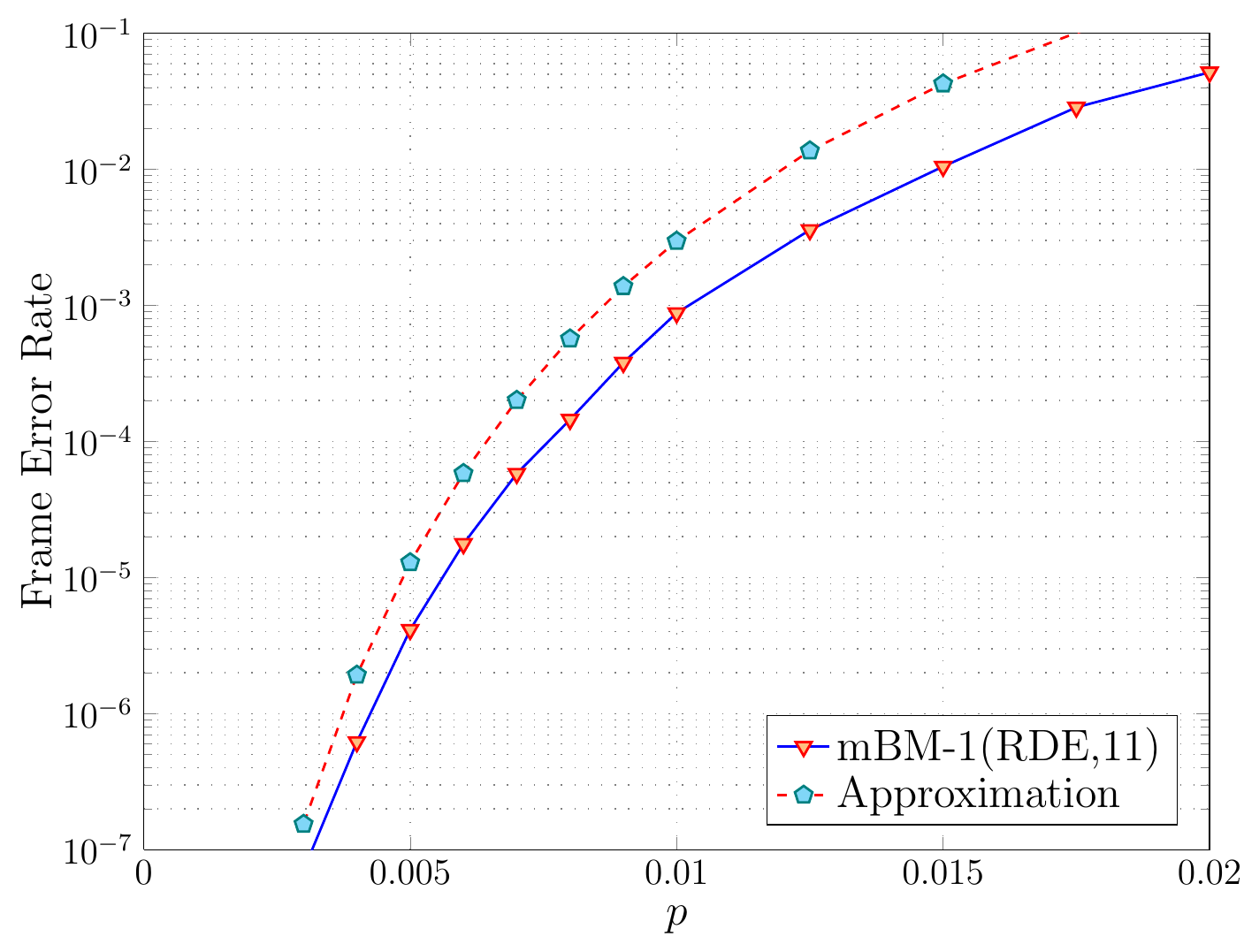}\caption{\label{fig:simqSC} Performance of mBM-1(RDE,11) and its approximation
$2^{-F}$ where $F$ is given in (\ref{eq:ComputeF}) for the (255,239)
RS code over an $m$-SC($p$) channel.}

\end{figure}

\begin{remrk}
In Fig. \ref{fig:simqSC}, we simulate the performance of mBM-1(RDE,11)
for the (255,239) RS code over an $m$-SC channel. One curve reflects
the simulated frame-error rate (FER) and the other is the approximation
derived from $2^{-F}$ where $F$ is given in (\ref{eq:ComputeF})
with $R=11$.
\end{remrk}

\section{Some Extensions \label{sec:Ext-and-Gen}}

\subsection{Erasure patterns using covering codes}

The RD framework we use is most suitable when $N\rightarrow\infty$.
For a finite $N$, choosing random codes for only a few LRPs can be
risky. We can instead use good covering codes to handle these LRPs.
In the scope of covering problems, one can use an $\ell$-ary $t_{c}$-covering
code (e.g., a perfect Hamming or Golay code) with covering radius
$t_{c}$ to cover the whole space of $\ell$-ary vectors of the same
length. The covering may still work well if the distortion measure
is close to, but not exactly equal to the Hamming distortion. The
method of using covering codes in the LRPs was proposed earlier in
\cite{Tokushige-ieice03} to choose the test patterns in iterative
bounded distance decoding algorithms for binary linear block codes.

In order take care of up to the $\ell$ most likely symbols at each
of the $n_{c}$ LRPs of an $(N,K)$ RS, we consider an $(n_{c},k_{c})$
$\ell$-ary $t_{c}$-covering code whose codeword alphabet is $\mathbb{Z}_{\ell+1}\setminus\{0\}=\{1,2,\ldots,\ell\}.$
Then, we give a definition of the (generalized) error patterns and
erasure patterns for this case. In order to draw similarities between
this case and the previous cases, we still use the terminology {}``generalized
erasure pattern'' and shorten it to erasure pattern even if errors-only
decoding is used. For errors-only decoding, Condition \ref{con:BMAerr-n-era}\emph{
}for successful decoding becomes\[
\nu<\frac{1}{2}(N-K+1).\]

\begin{definitn}
(Error patterns and erasure patterns for errors-only decoding) Let
us define $x^{N}\in\mathbb{Z}_{\ell+1}^{N}$ as an error pattern where,
at index $i$, $x_{i}=j$ implies that the $j$-th most likely symbol
is correct for $j\in\{1,2,\ldots\ell\}$, and $x_{i}=0$ implies none
of the first $\ell$ most likely symbols is correct. Let $\hat{x}^{N}\in\{1,2,\ldots,\ell\}^{N}$
be an erasure pattern where, at index $i$, $\hat{x}_{i}=j$ implies
that the $j$-th most likely symbol is chosen as the hard-decision
symbol for $j\in\{1,2,\ldots,\ell\}$.\end{definitn}
\begin{prop}
If we choose the \emph{letter-by-letter} distortion measure $\delta:\mathbb{Z}_{\ell+1}\times\mathbb{Z}_{\ell+1}\setminus\{0\}\rightarrow\mathbb{R}_{\geq0}$
defined by $\delta(x,\hat{x})=[\Delta]_{x,\hat{x}}$ in terms of the
$(\ell+1)\times\ell$ matrix\begin{equation}
\Delta=\left(\begin{array}{cccc}
1 & 1 & \ldots & 1\\
0 & 1 & \ldots & 1\\
1 & 0 & \ldots & 1\\
\vdots & \vdots & \ddots & \vdots\\
1 & 1 & \ldots & 0\end{array}\right)\label{eq:PFdst}\end{equation}
then the condition for successful errors-only decoding then becomes\begin{equation}
d(x^{N},\hat{x}^{N})<\frac{1}{2}(N-K+1).\label{eq:SCPF}\end{equation}
\end{prop}
\begin{proof}
It follows directly from \[d(x^{N},\hat{x}^{N})=\sum_{k=1}^{\ell}\sum_{j=0,j\neq k}^{\ell}\chi_{j,k}=\nu.\]\end{proof}
\begin{remrk}
If we delete the first row which corresponds to the case where none
of the first $\ell$ most likely symbols is correct then the distortion
measure is exactly the Hamming distortion. 
\end{remrk}

\paragraph*{Split covering approach}

We can break an error pattern $x^{N}$ into two sub-error patterns
$x^{LRPs}\triangleq x_{\sigma(1)}x_{\sigma(2)}\ldots x_{\sigma(n_{c})}$
of $n_{c}$ least reliable positions and $x^{MRPs}\triangleq x_{\sigma(n_{c}+1)}\ldots x_{\sigma(N)}$
of $N-n_{c}$ most reliable positions. Similarly, we can break an
erasure pattern $\hat{x}^{N}$ into two sub-erasure patterns $\hat{x}^{LRPs}\triangleq\hat{x}_{\sigma(1)}\hat{x}_{\sigma(2)}\ldots\hat{x}_{\sigma(n_{c})}$
and $\hat{x}^{MRPs}\triangleq\hat{x}_{\sigma(n_{c}+1)}\ldots\hat{x}_{\sigma(N)}$.
Let $z_{n_{c}}$ be the number of positions in the $n_{c}$ LRPs where
none of the first $\ell$ most likely symbols is correct, or \[z_{n_{c}}=\left|\left\{ i=1,2,\ldots,n_{c}:x_{\sigma(i)}=0\right\} \right|.\]
If we assign the set of all sub-error patterns $\hat{x}^{LRPs}$ to
be an $(n_{c},k_{c})$ $t_{c}$-covering code then \[d(x^{LRPs},\hat{x}^{LRPs})\leq t_{c}+z_{n_{c}}\]
because this covering code has covering radius $t_{c}$. Since \[d(x^{N},\hat{x}^{N})=d(x^{LRPs},\hat{x}^{LRPs})+d(x^{MRPs},\hat{x}^{MRPs}),\]
in order to increase the probability that the condition (\ref{eq:SCPF})
is satisfied we want to make $d(x^{MRPs},\hat{x}^{MRPs})$ as small
as possible by the use of the RD approach. The following proposition
summarizes how to generate a set of $2^{R}$ erasure patterns for
multiple runs of errors-only decoding.
\begin{prop}
In each erasure pattern, the letter sequence at $n_{c}$ LRPs is set
to be a codeword of an $(n_{c},k_{c})$ $\ell$-ary $t_{c}-$covering
code. The letter sequence of the remaining $N-n_{c}$ MRPs is generated
randomly by the RD method (see Section \ref{sec:Proposed-Algorithm})
with rate $R_{MRPs}=R-k_{c}\log_{2}\ell$ and the distortion measure
in (\ref{eq:PFdst}). Since this covering code has $\ell^{k_{c}}$
codewords, the total rate is $R_{MRPs}+\log_{2}\ell^{k_{c}}=R.$\end{prop}
\begin{example}
For a (7,4,3) binary Hamming code which has covering radius $t_{c}=1$,
we take care of the $2$ most likely symbols at each of the 7 LRPs.
We see that $1001001$ is a codeword of this Hamming code and then
form erasure patterns $1001001\hat{x}_{8}\hat{x}_{9}\ldots\hat{x}_{n}$
with assumption that the positions are written in increasing reliability
order. The $2^{R-4}$ sub-erasure patterns $\hat{x}_{8}\hat{x}_{9}\ldots\hat{x}_{n}$
are generated randomly using the RD approach with rate $(R-4)$.
\end{example}

\begin{remrk}
While it also makes sense to use a covering codes for the $n_{c}$
LRPs of the erasure patterns and set the rest to be letter $1$ (i.e.,
chose the most likely symbol as the hard-decision), our simulation
results shows that the performance can usually be improved by using
a combination of a covering code and a random (i.e., generated by
the RD approach) code. More discussions are presented in Section \ref{sec:Simulation-results}.
\end{remrk}

\subsection{A single decoding attempt}

In this subsection, we investigate a special case of our proposed
RDE framework when $R=0$ (i.e., the set of erasure patterns consists
of one pattern). In this case, our proposed approach is related to
another line of work where one tries to design a good erasure pattern
for a single BM decoding or a good multiplicity matrix for a single
ASD decoding \cite{Parvaresh-isit03,Ratnakar-it05,El-Khamy-dimacs05,Das-isit09}.
We will see that the RDE approach for $R=0$ is quite similar to optimizing
a Chernoff bound \cite{Ratnakar-it05,El-Khamy-dimacs05} or using
the method of types \cite{Das-isit09}. The main difference is that
this approach starts from Condition \ref{con:ASDcond} rather than
its large multiplicity approximation.

\begin{lemma}
\label{lem:Rate0} When rate $R=0$, the distribution matrix $\mathbf{Q}$
that optimizes the RDE/RD function consists of only binary entries.
Consequently, the random codebook using the proposed RDE approach
(the set of erasure patterns) becomes a single deterministic pattern. \end{lemma}
\begin{proof}
[Sketch of proof]For each $(s,t)$ pair, the total rate is the sum
of $N$ individual components as seen in Proposition \ref{pro:FactoredRDE}.
Therefore, the zero total rate implies all components are zero. Thus,
it suffices to show that if an arbitrary rate component (denoted as
$R$ in the proof) is zero then the corresponding column of $\mathbf{Q}$
has all entries equal to $0$ or $1$. 

For the RD case, it is well known \cite[p. 27]{Berger-1971} that
if $R=0$ then the distortion is given by $D_{\max}=\min_{k}\sum_{j}p_{j}\delta_{jk}$
where $k^{\star}$ is the argument that achieves this minimum and
the test-channel input distribution is \[
q_{k}^{\star}=\begin{cases}
1 & \mbox{if }k=k^{\star}\\
0 & \mbox{otherwise}\end{cases}.\]
Computing the RDE for the source distribution $p_{j}$ is equivalent
to solving the RD problem for an appropriately tilted source distribution
$\tilde{p}_{j}^{\star}$. Therefore, the above property is inherited
by the RDE as well. In particular, the distortion at $R=0$ is given
by $\min_{k}\sum_{j}\tilde{p}_{j}^{\star}\delta_{jk}$ and the test-channel
input distribution is supported on the singleton element that achieves
this minimum. 

This result can also be shown directly by solving (\ref{eq:RDEmaxmin})
while dropping the rate constraint from (\ref{eq:RDEmaxminPRD}).
\end{proof}
Let $G_{k}(D)$ be the large deviation rate-function for the distortion
when the reconstruction symbol is fixed to $k$. It is well-known
that this can be computed using either a Chernoff bound or the method
of types \cite{Cover-1991}. Both techniques result in the same function;
for $\alpha\geq0$, it is described implicitly by\begin{align*}
D(\alpha) & =\frac{\sum_{j}p_{j}2^{\alpha\delta_{j,k}}\delta_{j,k}}{\sum_{j'}p_{j'}2^{\alpha\delta_{j',k}}},\\
G_{k}(\alpha) & =\sum_{j}\frac{p_{j}2^{\alpha\delta_{j,k}}}{\sum_{j'}p_{j'}2^{\alpha\delta_{j',k}}}\log\frac{2^{\alpha\delta_{j,k}}}{\sum_{j'}p_{j'}2^{\alpha\delta_{j',k}}}.\end{align*}

\begin{thm}
The RDE function for $R=0$ is equal to\[
F(0,D)=\max_{k}G_{k}(D).\]
\end{thm}

\begin{proof}
Lemma \ref{lem:Rate0} shows that the reconstruction distribution
must be supported on a single element. Since the exponential failure
probability for any fixed reconstruction symbol follows from a standard
large-deviations analysis, the only remaining degree of freedom is
which symbol to use. Choosing the best symbol maximizes the RDE.\end{proof}

\begin{remrk}
This means that the single decoding attempt with the best error-exponent
can be computed as a special case of the RDE approach. Simplifying
our proposed algorithm to use the single Lagrange multiplier $\alpha$
leads to an algorithm that is very similar to the one proposed in
\cite{Das-isit09}. It also seems unlikely that this new algorithm
will provide any significant performance gains either in performance or
complexity. 

\end{remrk}

\section{Simulation results\label{sec:Simulation-results}}
\begin{figure*}[t!]
\centering{}\includegraphics[scale=1]{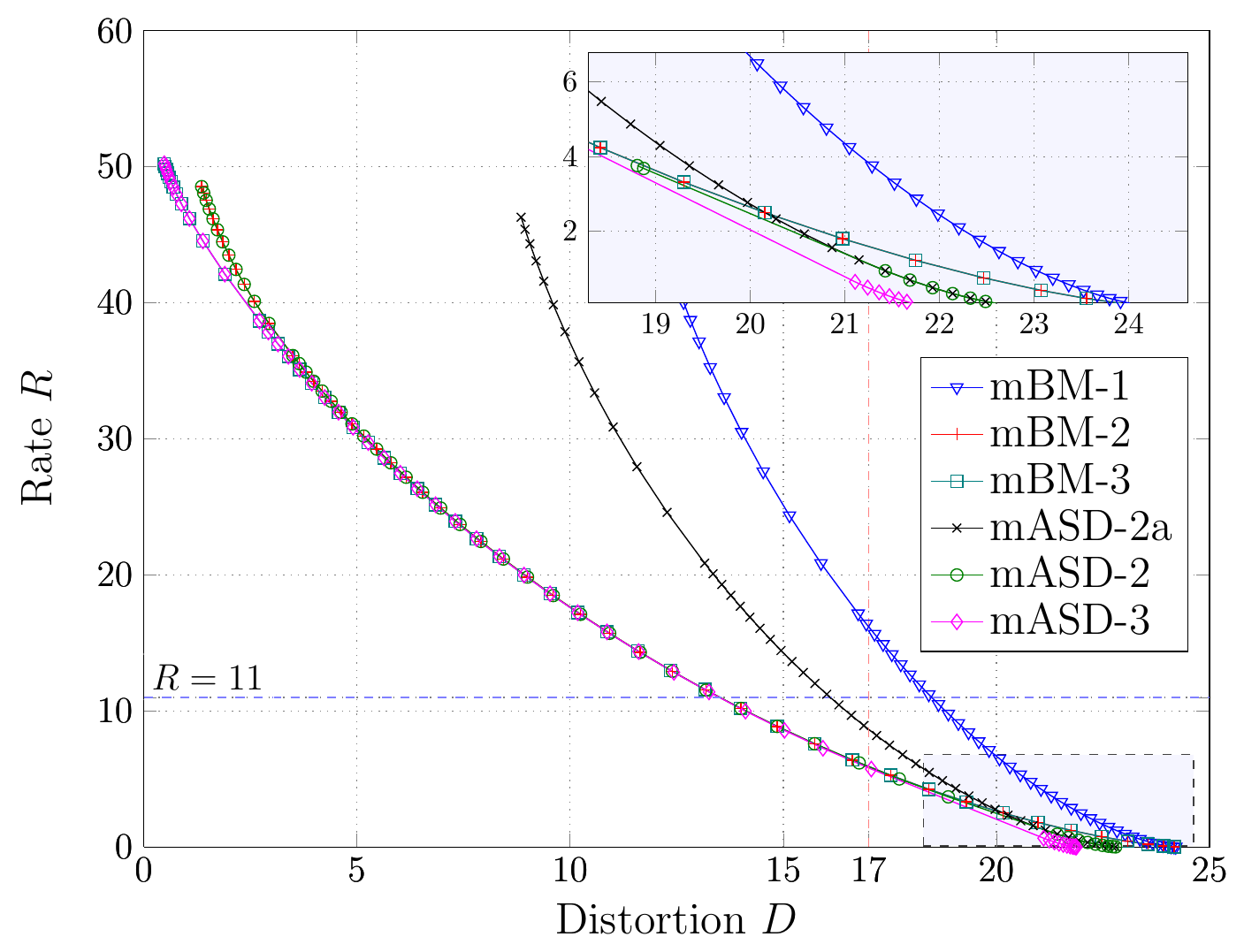}\caption{\label{fig:rdcurve} A realization of RD curves at $E_{b}/N_{0}=5.2$
dB for various decoding algorithms for the (255,239) RS code over
an AWGN channel. }

\end{figure*}

\begin{figure}[t!]
\centering{}\includegraphics[scale=1]{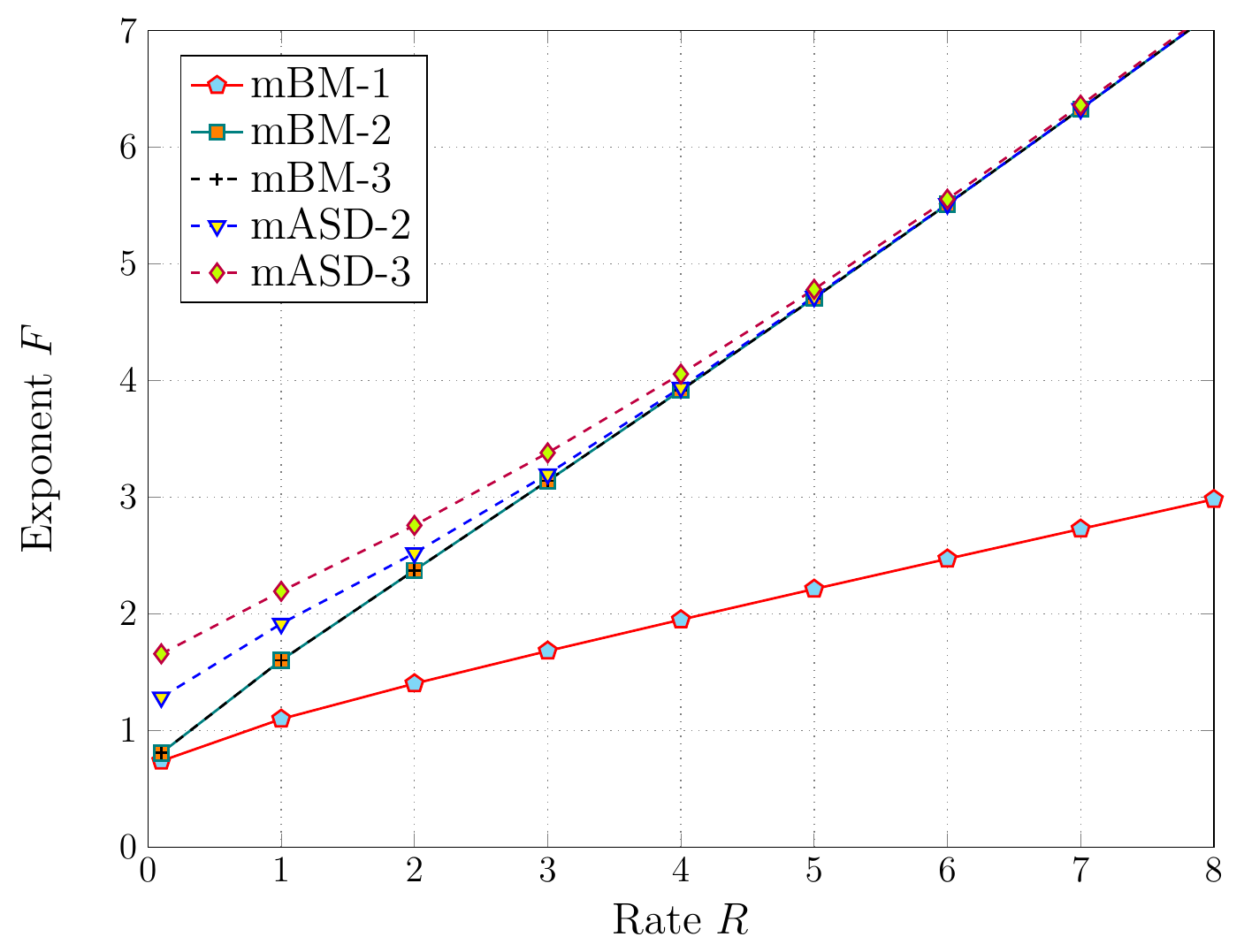}\caption{\label{fig:rdecurve} A realization of RDE curves at $E_{b}/N_{0}=6$
dB for various decoding algorithms for the (255,239) RS code over
an AWGN channel.}

\end{figure}

\begin{figure*}[t!]
\centering{}\includegraphics[clip,scale=1]{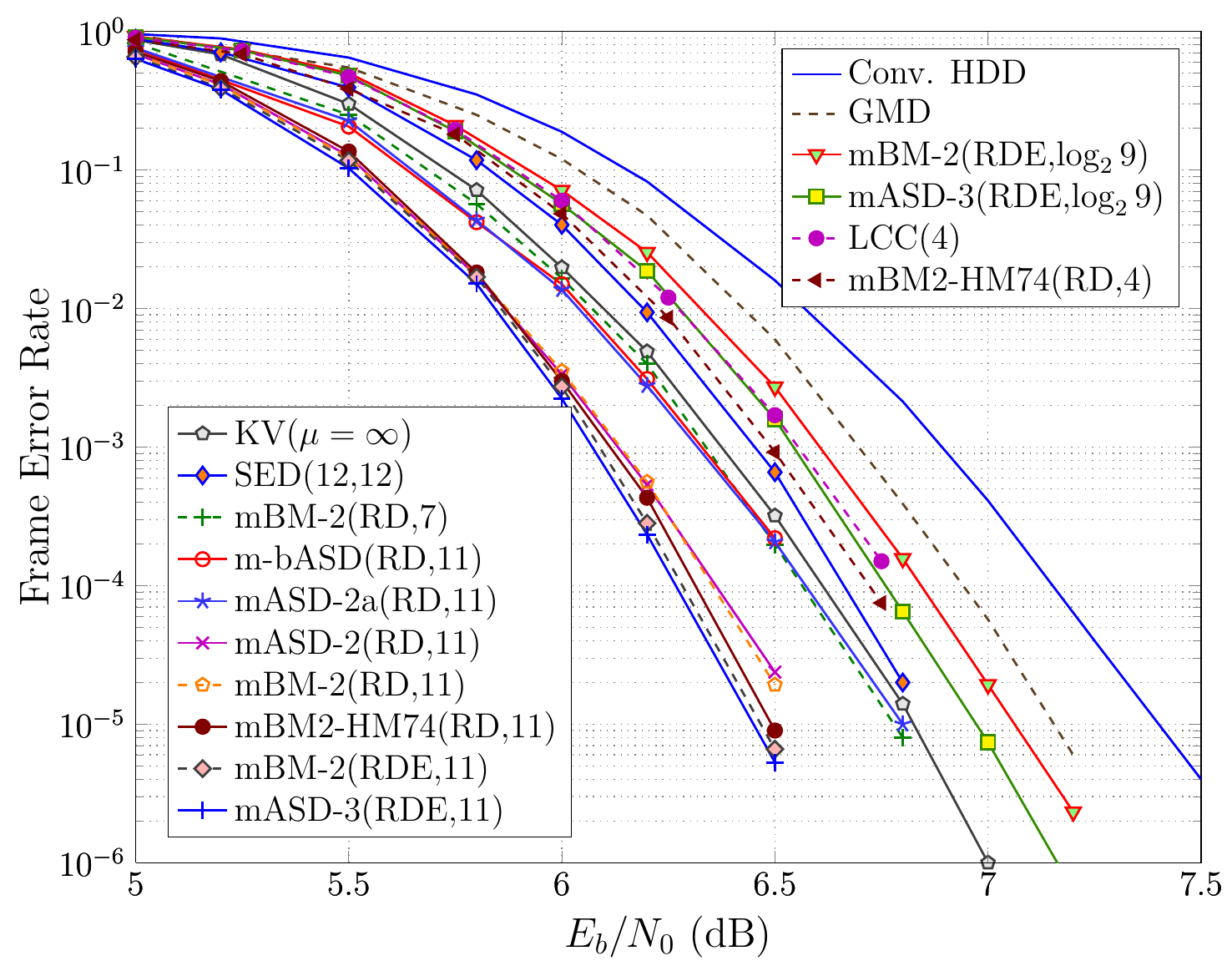}\caption{\label{fig:berRS239} Performance of various decoding algorithms for
the (255,239) RS code using BPSK over an AWGN channel.}
\end{figure*}

\begin{figure}
\centering{}\includegraphics[clip,scale=1]{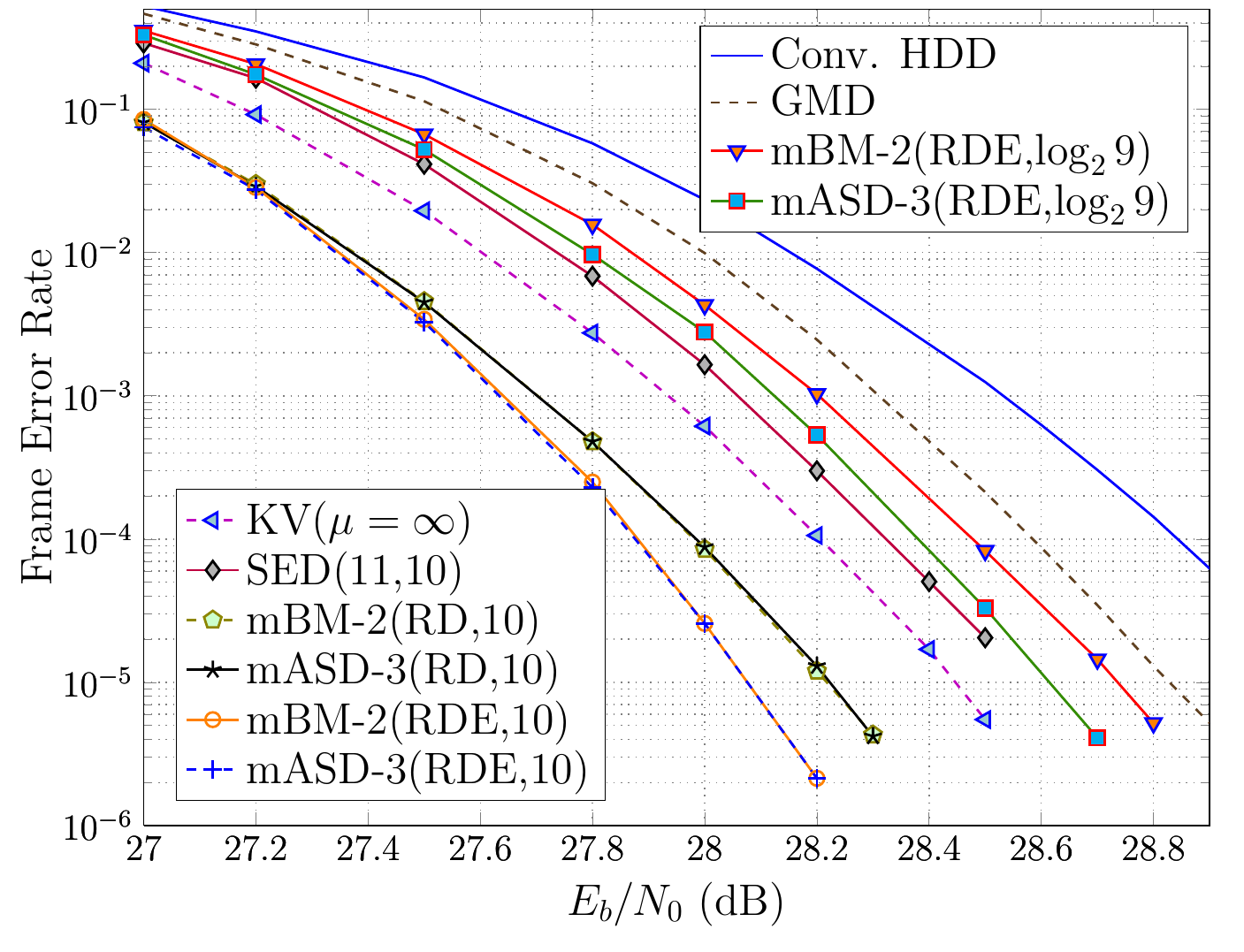}\caption{\label{fig:berRS239-QAM} Performance of various decoding algorithms
for the (255,239) RS code using 256-QAM over an AWGN channel.}
\end{figure}
\begin{figure*}[t]
\centering{}\includegraphics[scale=1]{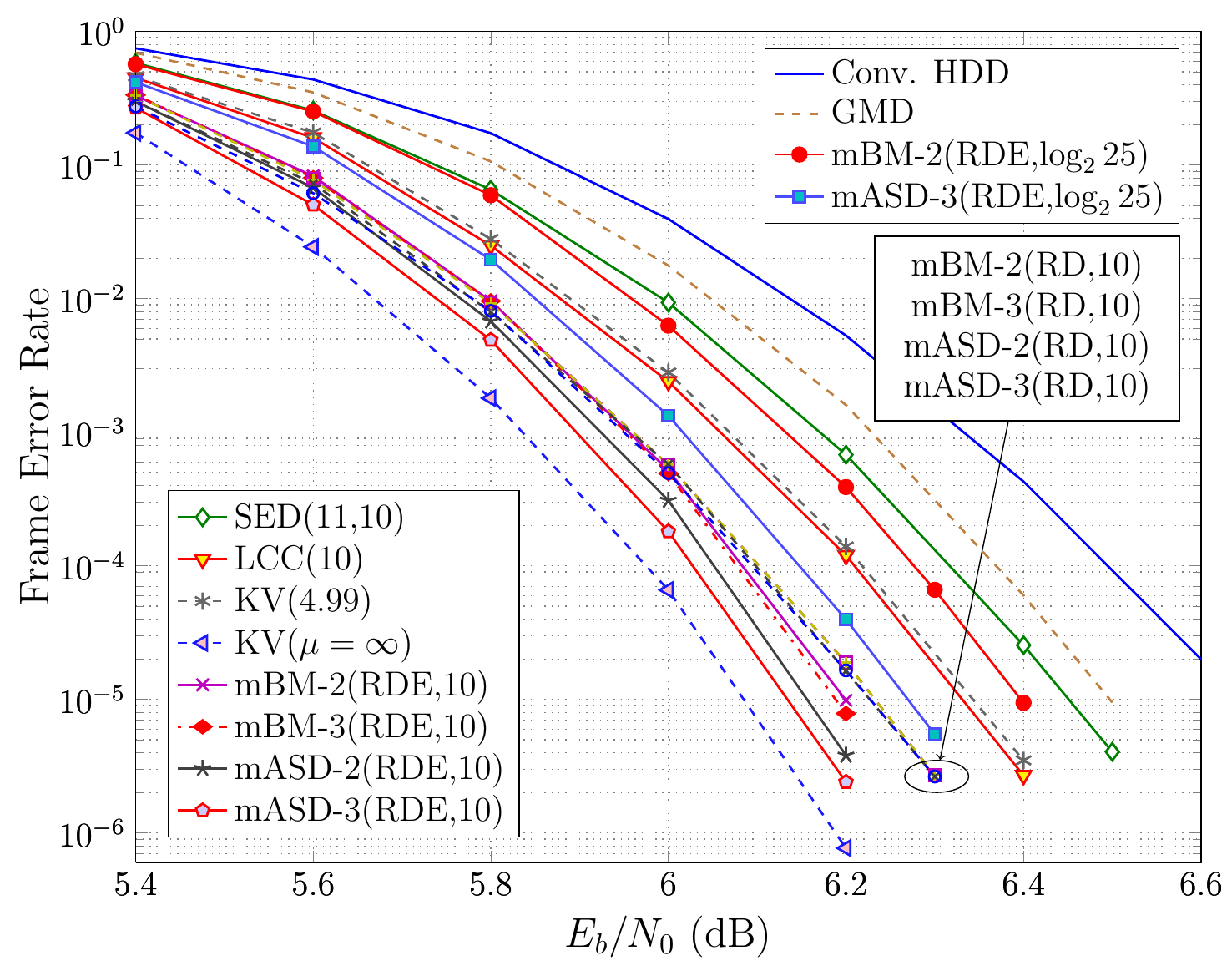}\caption{\label{fig:beRS458} Performance of various decoding algorithms for
the (458,410) RS code over $\mathbb{F}_{2^{10}}$ using BPSK over
an AWGN channel.}
\end{figure*}
\begin{figure}[t]
\centering{}\includegraphics[scale=1]{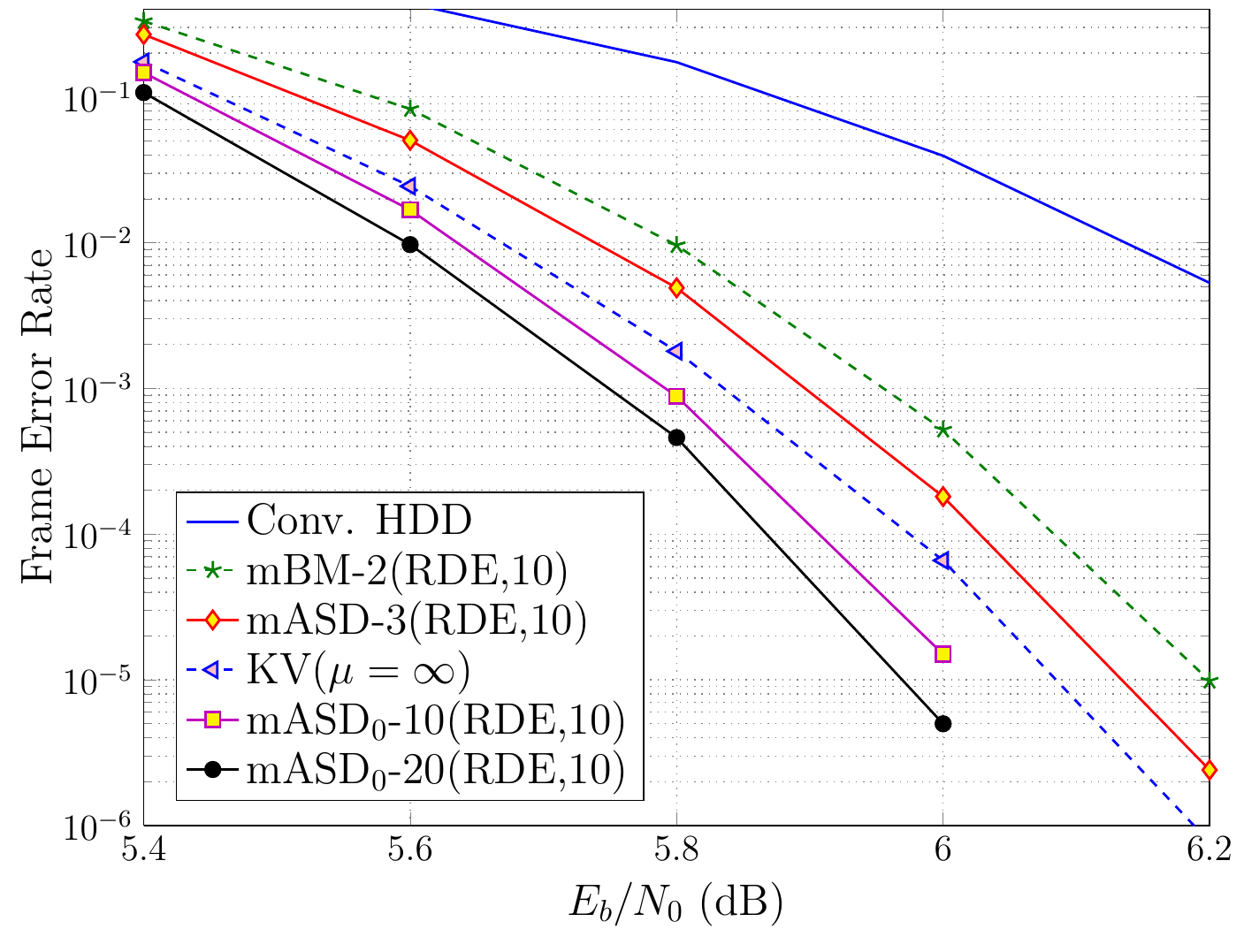}\caption{\label{fig:RS410asdM}Performance of various decoding algorithms for
the (458,410) RS code over $\mathbb{F}_{2^{10}}$ using BPSK over
an AWGN channel.}
\end{figure}
\begin{figure}[t]
\centering{}\includegraphics[scale=1]{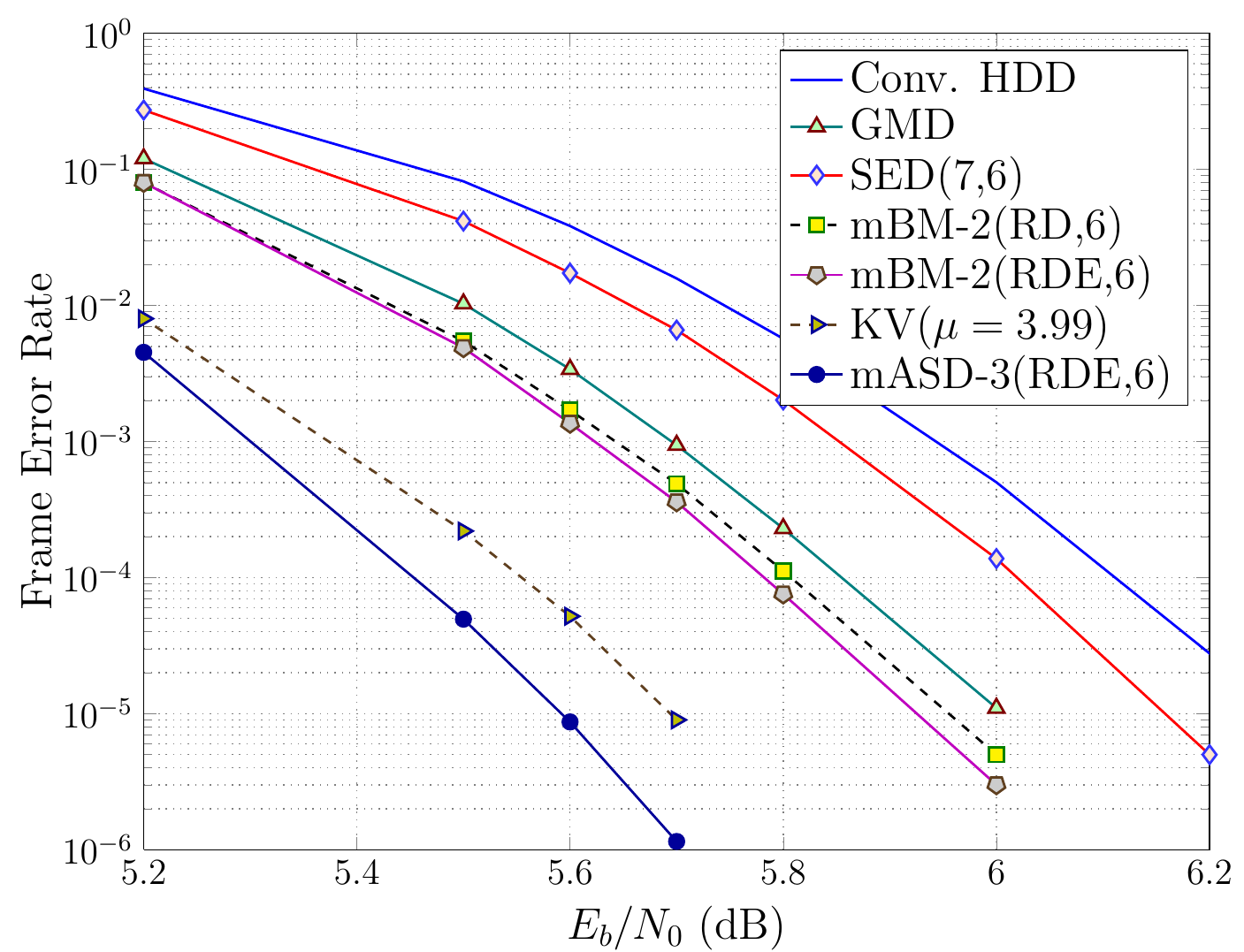}\caption{\label{fig:beRS127} Performance of various decoding algorithms for
the (255,127) RS code using BPSK over an AWGN channel.}
\end{figure}
\begin{figure}[t]
\centering{}\includegraphics[clip,scale=1]{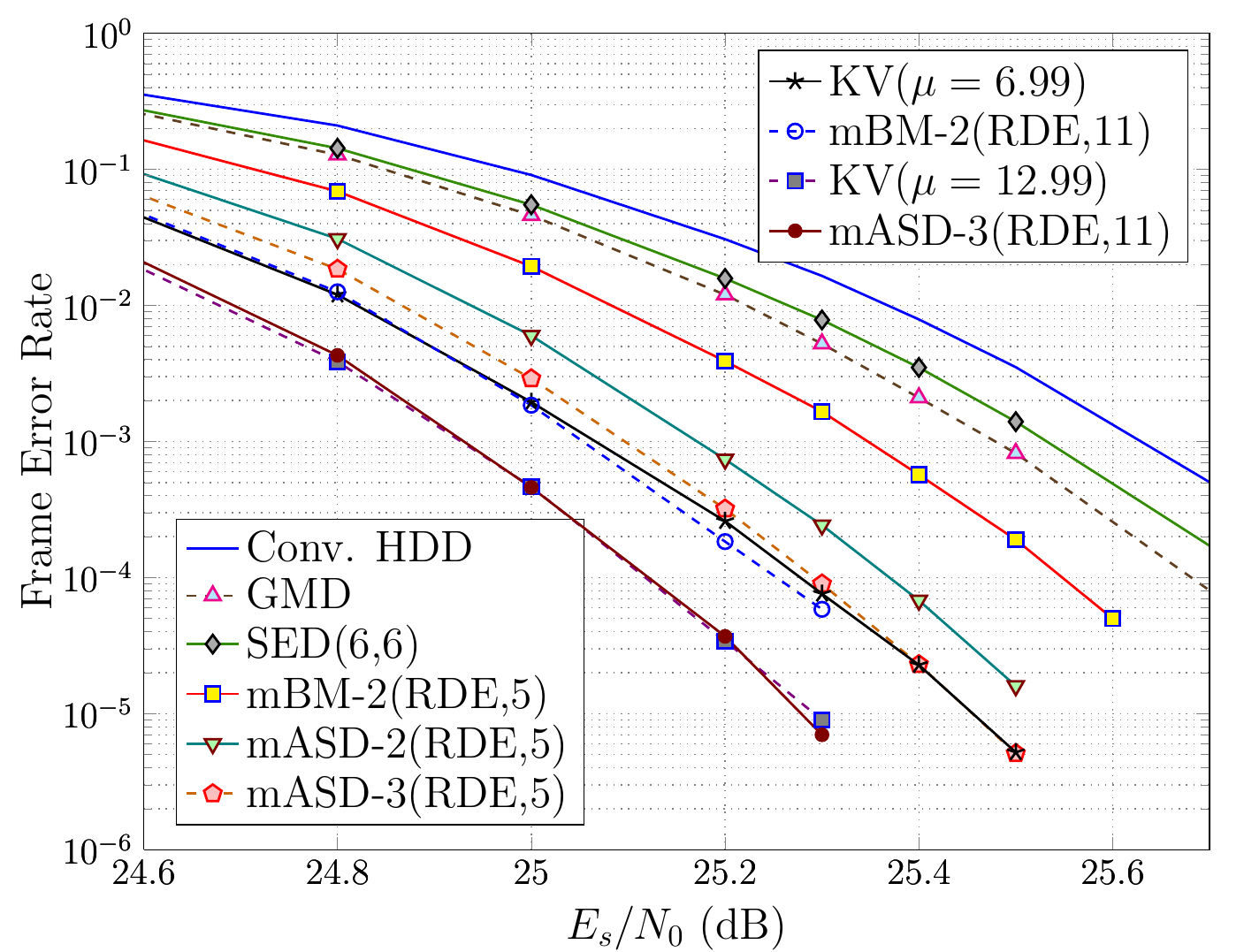}\caption{\label{fig:berRS191QAM} Performance of various decoding algorithms
for the (255,191) RS code using 256-QAM over an AWGN channel.}
\end{figure}
\begin{figure}[t]
\centering{}\includegraphics[clip,scale=1]{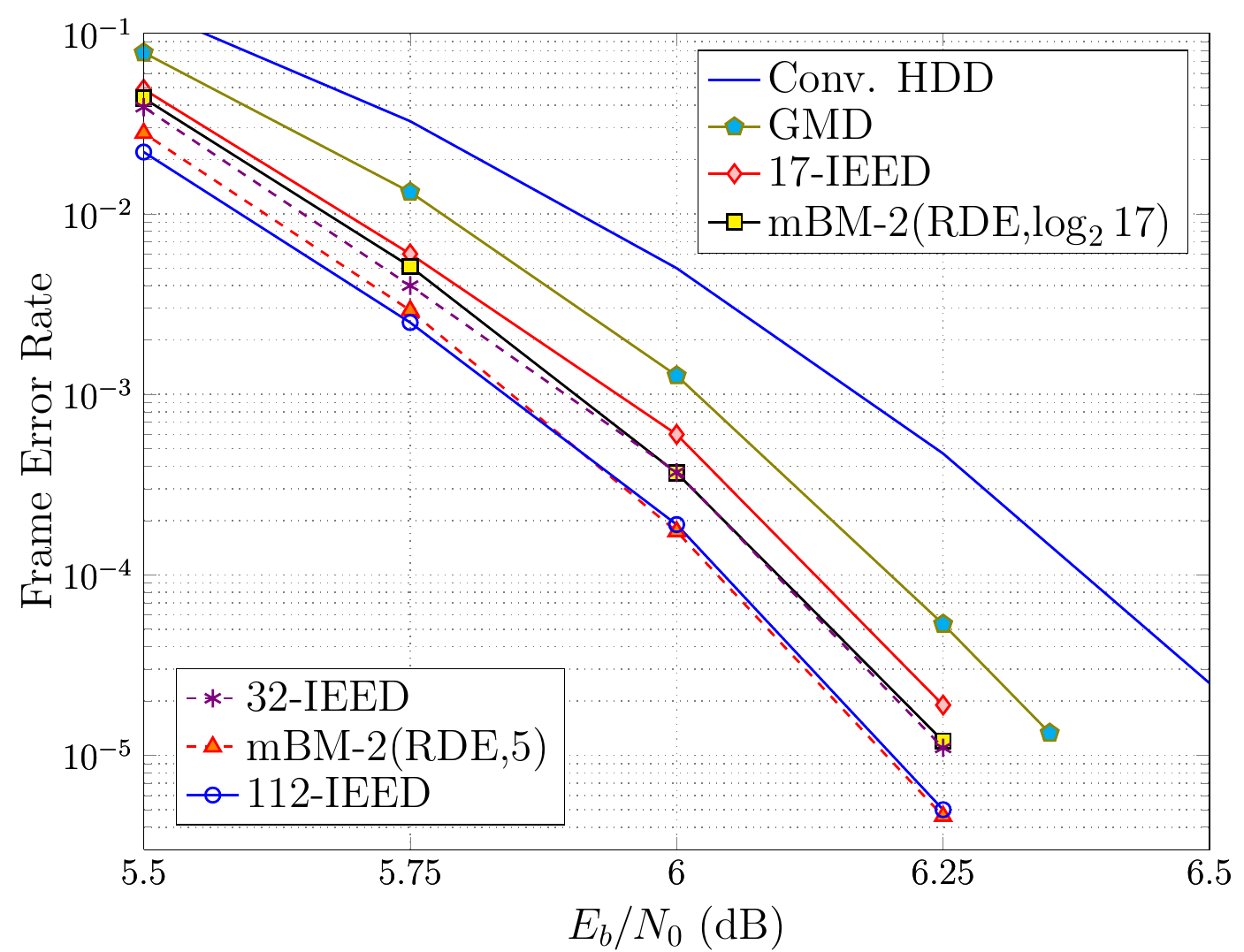}\caption{\label{fig:berRS223} Performance of various decoding algorithms for
the (255,223) RS code using BPSK over an AWGN channel.}
\end{figure}

In this section, we present simulation results on the performance
of RS codes over an AWGN channel with either BPSK or 256-QAM as the
modulation format. In all the figures, the curve labeled mBM-1 corresponds
to standard errors-and-erasures BM decoding with multiple erasure
patterns. For $\ell>1$, the curves labeled mBM-$\ell$ correspond
to errors-and-erasures BM decoding with multiple decoding trials using
both erasures and the top-$\ell$ symbols. The curves labeled mASD-$\mu$
correspond to multiple ASD decoding trials with maximum multiplicity
$\mu$. The number of decoding attempts is $2^{R}$ where $R$ is
denoted in parentheses in each algorithm's acronym (e.g., mBM-2(RD,11)
uses the RD approach with $R=11$ while mBM-2(RDE,10) uses the RDE
approach with $R=10$). Please note that not all the algorithms listed
in this section are of the same complexity unless stated explicitly.


In Fig. \ref{fig:rdcurve}, the RD curves are shown for various algorithms
using the RD approach at $E_{b}/N_{0}=5.2$ dB where BPSK is used.
For the (255,239) RS code, the fixed threshold for decoding is $D=N-K+1=17$.
Therefore, one might expect that algorithms whose average distortion
is less than $17$ should have a frame error rate (FER) less than
$\frac{1}{2}$. The RD curve allows one to estimate the number of
decoding patterns required to achieve this FER. Notice that the mBM-1
algorithm at rate $0$, which is very similar to conventional BM decoding,
has an expected distortion of roughly $24$. For this reason, the
FER for conventional decoding is close to $1$. The RD curve tells
us that trying roughly $2^{16}$ (i.e., $R=16$) erasure patterns
would reduce the FER to roughly $\frac{1}{2}$ because this is where
the distortion drops down to $17$. Likewise, the mBM-2 algorithm
using rate $R=11$ has an expected distortion of less than $14$.
So we expect (and our simulations confirm) that the FER should be
less than $\frac{1}{2}$. 

One weakness of this RD approach is that RD describes only the average
distortion and does not directly consider the probability that the
distortion is greater than $17$. Still, we can make the following
observations from the RD curve. Even at high rates (e.g., $R\ge5$),
we see that the distortion $D$ achieved by mBM-2 is roughly the same
as mBM-3, mASD-2, and mASD-3 but smaller than mASD-2a (see Example
\ref{exa:mASD2a}) and mBM-1. This implies that, for this RS code,
mBM-2 using the RD approach is no worse than the more complicated
ASD based approaches for a wide range of rates (i.e., $5\leq R\leq35$).
This is also true if the RDE approach is used as can be seen in Fig.
\ref{fig:rdecurve} which depicts the trade-off between rate R and
exponent F for various algorithms at $E_{b}/N_{0}=6$ dB. For this
RS code, ASD based approaches have a better exponent than mBM-2 at
low rates (i.e., small number of decoding trials) and have roughly
the same exponent for rates $R\geq5$.

In Fig. \ref{fig:berRS239}, a plot of the FER versus $E_{b}/N_{0}$
is shown for the (255,239) RS code over an AWGN channel with BPSK
as the modulation format. The conventional HDD and the GMD algorithms
have modest performance since they use only one or a few decoding
attempts. Choosing $R=11$ allows us to make fair comparisons with
SED(12,12). With the same number of decoding trials, mBM-2(RD,11)
outperforms SED(12,12) by $0.3$ dB at FER$=10^{-4}$. Even mBM-2(RD,7),
with many fewer decoding trials, outperforms both SED(12,12) and the
KV algorithm with $\mu=\infty$. Among all our proposed algorithms
using the RD approach with rate $R=11$, the mBM2-HM74(RD,11) achieves
the best performance. This algorithm uses the Hamming (7,4) covering
code for the $7$ LRPs and the RD approach for the remaining codeword
positions. Meanwhile, small differences in the performance among mBM-2(RD,11),
mBM-3(RD,11), mASD-2(RD,11), and mASD-3(RD,11) suggest that: (i) taking
care of the $2$ most likely symbols at each codeword position is
good enough for multiple decoding of this RS code and (ii) multiple
runs of errors-and-erasures decoding is generally almost as good as
multiple runs of ASD decoding. Recall that this result is also correctly
predicted by the RD analysis. When the RDE approach is used, mBM-2(RDE,11)
still has roughly the same performance as a more complex mASD-3(RDE,11).
One can also observe that these two algorithms using the RDE approach
achieve better performance than mBM-2(RD,11) and mBM2-HM74(RD,11)
that use the RD approach. We also simulate our proposed algorithm
at $R=\log_{2}9$ to compare with the GMD algorithm. While both mBM-2(RDE,$\log_{2}9$)
and the GMD algorithm use the same number of $9$ errors-and-erasures
decoding attempts, mBM-2(RDE,$\log_{2}9$) yields roughly a $0.1$
dB gain. The simulation results show that, at this low rate $R=\log_{2}9$,
mASD-3 has a larger gain over mBM-2 than at a higher rate $R=11$.
This phenomenon can be predicted in Fig. \ref{fig:rdecurve} where
mASD-3 starts to achieve a larger exponent $F$ at small values of
$R$.

To compare with the Chase-type approach (LCC) used in \cite{Bellorado-it10},
in Fig. \ref{fig:berRS239} we also consider the mBM2-HM74(4) algorithm
that uses the Hamming (7,4) covering code for the $7$ LRPs and the
hard decision pattern for the remaining codeword positions. This shows
that, for the (255,239) RS code, the mBM2-HM74 achieves better performance
than the LCC(4) with the same number ($2^{4}$) of decoding attempts.
For the (458,410) RS code considered in Fig. \ref{fig:beRS458}, one
can also observe that the group of algorithms that we propose have
better performance than LCC(10) with the same number ($2^{10})$ of
decoding attempts. However, the implementation complexity of LCC(10)
may be lower than the algorithms proposed here due to their clever
techniques that reduce the decoding complexity per trial. It is also
interesting to note that the method proposed here, based on covering
codes and random codebook generation, is also compatible with some
of the fast techniques used by the LCC decoding.

We also performed simulations using QAM and Fig. \ref{fig:berRS239-QAM}
shows FER versus $E_{b}/N_{0}$ performance of the same (255,239)
RS code transmitted over an AWGN channel with 256-QAM modulation.
At FER=$10^{-4}$, our proposed algorithms mBM-2(RD,10) and mBM-2(RDE,10)
achieve $0.3-0.4$ dB gain over SED(11,10) (with the same complexity)
and also outperform KV($\mu=\infty)$. At $R=10$, mBM-2 still achieves
roughly the same performance as mASD-3.

In Fig. \ref{fig:beRS458}, a plot of the FER versus $E_{b}/N_{0}$
is shown for the (458,410) RS code that has a longer block length.
In this plot, BPSK is used as the modulation format and we also focus
on rate $R=10$. With algorithms that use the RD approach, mBM-2(RD,10)
still has approximately the same performance as mBM-3(RD,10), mASD-2(RD,10),
mASD-3(RD,10). However, when the RDE approach is employed, algorithms
that run multiple ASD decoding attempts have a recognizable gain over
algorithms that use multiple runs of BM decoding. The performance
gain of the RDE approach (over the RD approach) is small, but can
be seen easily by comparing mASD-3(RDE,10) to mASD-3(RD,10). As a
reference, we also plot the performance of KV($4.99$) which corresponds
to the proportional KV algorithm \cite{Gross-com06} with the scaling
factor $4.99$.

In Fig. \ref{fig:RS410asdM}, the same setting is used as in Fig.
\ref{fig:beRS458}. As can be seen in the figure, KV($\mu=\infty$)
achieve better performance than mASD-3(RDE,10) and mBM-2(RDE,10).
However, by considering higher $\mu$, our algorithms using the heuristic
method mASD$_{0}$-10(RDE,10) and mASD$_{0}$-20(RDE,10) can outperform
KV($\mu=\infty$). 

To target RS codes of lower rate, we also ran simulations of the (255,127)
RS code over an AWGN channel with BPSK modulation and the results
can be seen in Fig. \ref{fig:beRS127}. While mBM-2(RDE,6), mBM-2(RD,6),
SED(7,6) and GMD all use the same number of about $64$ errors-and-erasures
decoding attempts, our proposed mBM-2 algorithms outperforms the other
two algorithms. As seen in the plot, mASD-3(RDE,6) has quite a large
gain over mBM-2(RD,6) which is reasonable since ASD decoding is known
to perform very well compared to BM decoding with low-rate RS codes.
In this figure, KV($3.99$) denotes the proportional KV algorithm
\cite{Gross-com06} with the scaling factor $3.99$ and therefore
with maximum multiplicity $\mu=3$. While mASD-3(RDE,6) with $64$
decoding attempts outperforms KV($3.99$) as expected, the small gain
of roughly $0.5$ dB at FER=$10^{-4}$ suggests that with low-rate
RS codes, one might prefer increasing $\mu$ in a single ASD decoding
attempt to running multiple ASD decoding attempts of a lower $\mu$.

In Fig. \ref{fig:berRS191QAM}, we show the FER versus $E_{s}/N_{0}$
performance for the (255,191) RS codes using 256-QAM. Again, our proposed
algorithm mBM-2(RDE,5) performs favorably compared to SED(6,6) and
GMD with the same number of about $32$ errors-and-erasures decoding
attempts. Under this setup, mASD-2(RDE,5) and mASD-3(RDE,5) achieve
significant gains over mBM-2(RDE,5). Our proposed mASD-3(RDE,11) and
mASD-3(RDE,5) algorithms have fairly the same performance as the proportional
KV algorithm with the scaling factor $12.99$ and $6.99$, respectively.

To compare with the iterative erasure and error decoding (IEED) algorithm
proposed in \cite{Tokushige-ieice06}, we also conducted simulations
of the (255,223) RS code over an AWGN channel using BPSK and the results
are shown in Fig. \ref{fig:berRS223}. With the same number of about
$17$ errors-and-erasures decoding attempts, our proposed mBM-2(RDE,$\log_{2}17$)
algorithm outperforms both the GMD and 17-IEED algorithms. In fact,
at FER smaller than $10^{-3}$, mBM-2(RDE,$\log_{2}17$) has roughly
the same performance as 32-IEED which needs to use $32$ decoding
attempts. Meanwhile, mBM-2(RDE,5) that uses $32$ decoding attempts
performs as good as 112-IEED where $112$ decoding attempts are required.

\section{Conclusion\label{sec:Conclusion}}

A unified framework based on rate-distortion (RD) theory has been
developed to analyze multiple decoding trials, with various algorithms,
of RS codes in terms of performance and complexity. An important contribution
of this paper is the connection that is made between the complexity
and performance (in an asymptotic sense) of these multiple-decoding
algorithms and the rate-distortion of an associated RD problem. Based
on this analysis, we propose two solutions; the first is based on
the RD function and the second on the RD exponent (RDE). The RDE analysis
shows that this approach has several advantages. Firstly, the RDE
approach achieves a near optimal performance-versus-complexity trade-off
among algorithms that consider running a decoding scheme multiple
times (see Remark \ref{rem:advantage}). Secondly, it helps estimate
the error probability using exponentially tight bounds for $N$ large
enough. Further, we have shown that covering codes can also be combined
with the RD approach to mitigate the suboptimality of random codes
when the effective block-length is not large. As part of this analysis,
we also present numerical and analytical computations of the RD and
RDE functions for sequences of i.n.d. sources. Finally, the simulation
results show that our proposed algorithms based on the RD and RDE
approaches achieve a better performance-versus-complexity trade-off
than previously proposed algorithms. One key result is that, for the
$(255,239)$ RS code, multiple-decoding using the standard Berlekamp-Massey
algorithm (mBM) is as good as multiple-decoding using more complex
algebraic soft-decision algorithms (mASD). However, for the $(458,410)$
RS code, the RDE approach improves the performance of mASD algorithms
beyond that of mBM decoding. 

Simulations results suggest an interesting
conjecture that for moderate-rate RS codes, multiple ASD decoding
attempts with small $\mu$ is preferred while for low-rate RS codes,
a single ASD decoding with large $\mu$ may be preferred. This conjecture
remains open for future research. Our future work will also focus on extending this framework to analyze
multiple decoding attempts for intersymbol interference channels.
In this case, it will be appropriate for the decoder to consider
multiple candidate error-events during decoding. Extending the RD
and RDE approaches directly to this case is not straightforward since
computing the RD and RDE functions for Markov sources in the large
distortion regime is still an open problem. Another interesting extension
is to use clever techniques to reuse the computations from one stage
of errors-and-erasures decoding to the next in order to lower the
complexity per decoding trial (e.g., \cite{Bellorado-it10}). 

\appendices

\section{Proof of Corollary \ref{cor:For-mASD}\label{sec:AppCorDmax}}
\begin{proof}
Using the formula in \cite[p. 27]{Berger-1971}, we have\[
D_{\max}=\sum_{i=1}^{N}\min_{k}\sum_{j=0}^{\ell}p_{i,j}\delta_{jk}.\]

For mBM-$\ell$ with distortion matrix in (\ref{eq:dstmBMl}), we
have $\sum_{j=0}^{\ell}p_{i,j}\delta_{jk}=\sum_{j\neq k}2p_{i,j}=2(1-p_{i,k})$
for $k\geq1$ and $\sum_{j=0}^{\ell}p_{i,j}\delta_{j0}=\sum_{j=0}^{\ell}p_{i,j}=1$.
Therefore,\begin{align*}
&D_{\max}(\mbox{mBM-}\ell) = \sum_{i=1}^{N}\min_{k=1,\ldots\ell}\{1,2(1-p_{i,k})\}\\
&\phantom{D_{\max}(\mbox{mBM-}\ell)} = \sum_{i=1}^{N}\min\{1,2(1-p_{i,1})\}
\end{align*}
since $p_{i,1}=\max_{k\geq1}\{p_{i,k}\}$.

Similarly, for mASD-$\mu$ with distortion matrix $\Delta_{\mu}$
in (\ref{eq:dstmASDmu_allrate}), we have
\begin{align*}
&\sum_{j=0}^{\ell}p_{i,j}\delta_{jk} = p_{i,0}\rho_{k,\mu}+\sum_{j=1}^{\ell}p_{i,j}\left(\rho_{k,\mu}-\frac{2m_{j,k}}{\mu}\right)\\
&\phantom{\sum_{j=0}^{\ell}p_{i,j}\delta_{jk}}= \rho_{k,\mu}-\sum_{j=1}^{\ell}\frac{m_{j,k}}{\mu}p_{i,j}
\end{align*}
for $k=1,\ldots,T$. Since multiplicity type 1 is always defined
to be $(\mu,0,\ldots,0)$, we have $\rho_{1,\mu}=2$ and consequently,
\[\sum_{j=0}^{\ell}p_{i,j}\delta_{j1}=2(1-p_{i,1}).\]
Therefore, we obtain
\begin{align*}
\begin{split}
&D_{\max}(\mbox{mASD-}\mu)=\sum_{i=1}^{N}\min_{k=2,\ldots,T}\left\{ 2(1-p_{i,1}),\rho_{k,\mu}-\sum_{j=1}^{\ell}\frac{m_{j,k}}{\mu}p_{i,j}\right\}.	 
\end{split}
\end{align*}
If mASD-$\mu$ uses multiplicity type $(0,0,\ldots0$) which is, for
example, labeled as type $T$ then we have \[\rho_{T,\mu}-\sum_{j=1}^{\ell}\frac{m_{j,T}}{\mu}p_{i,j}=\rho_{T,\mu}=1.\]
Consequently,
\begin{align*}
 & D_{\max}(\mbox{mASD-}\mu) = \sum_{i=1}^{N}\min_{k=2,\ldots,T-1}\left\{ 1,2(1-p_{i,1}),\rho_{k,\mu}-\sum_{j=1}^{\ell}\frac{m_{j,k}}{\mu}p_{i,j}\right\}\\
 &\phantom{D_{\max}(\mbox{mASD-}\mu)} \leq \sum_{i=1}^{N}\min\{1,2(1-p_{i,1})\}\\
 &\phantom{D_{\max}(\mbox{mASD-}\mu)} = D_{\max}(\mbox{mBM-}\ell)
\end{align*}
and this completes the proof.
\end{proof}

\section{Proof of Lemma \ref{lem:OneVar}\label{sec:AppLemRDmBM1}}
\begin{proof}
With the notation $\bar{p}=1-p$, according to \cite[p. 27]{Berger-1971} we have
\begin{align*}
&D_{\min} = \bar{p}\min_{k}\delta_{0k}+p\min_{k}\delta_{1k}=1-p\\
&D_{\max} = \min_{k}(\bar{p}\delta_{0k}+p\delta_{1k})=\min\{1,2(1-p)\}.\end{align*}

The function $R(D)$ is not defined for $D<D_{\min}$ and $R(D)=0$
for $D\geq D_{\max}$. For the case $D_{\min}\leq D<D_{\max}$, the
rate-distortion function $R(D)$ is given by solving the following
convex optimization problem

\begin{equation*}
\begin{array}{cc}
\min_{\mathbf{w}} & I(X;\hat{X})\\
\mbox{subject to} & w_{k|j}\triangleq\Pr(\hat{X}=k|X=j)\geq0\,\,\,\forall j,k\in\{0,1\}\\
 & w_{0|0}+w_{1|0}=1\\
 & w_{0|1}+w_{1|1}=1\\
 & \bar{p}w_{0|0}+pw_{0|1}+2\bar{p}w_{1|0}=D\end{array}
\end{equation*}
where the mutual information \[
I(X;\hat{X})=\bar{p}\sum_{k}w_{k|0}\log\frac{w_{k|0}}{q_{k}}+p\sum_{k}w_{k|1}\log\frac{w_{k|1}}{q_{k}}\]
 and the test-channel input probability-distribution\[
q_{k}=\Pr(\hat{X}=k)=\bar{p}w_{k|0}+pw_{k|1}.\]

We then form the Lagrangian\begin{align*}
&J(W) = I(X;\hat{X})+\sum_{j}\gamma_{j}(w_{0|j}+w_{1|j}-1) +\gamma(\bar{p}w_{0|0}+pw_{0|1}+2\bar{p}w_{1|0}-D)-\sum_{j,k}\lambda_{jk}w_{k|j}\end{align*}

and the Karush-Kuhn-Tucker (KKT) conditions become%
\footnote{Here we use some abuse of notation and still write the optimizing
values in their old forms without a $^{\star}$ notation. %
}\[
\begin{cases}
\frac{\partial J}{\partial w_{k|j}}=0 & \forall j,k\in\{0,1\}\\
w_{0|j}+w_{1|j}-1=0 & \forall j\in\{0,1\}\\
w_{k|j},\lambda_{jk}\geq0 & \forall j,k\in\{0,1\}\\
\lambda_{jk}w_{k|j}=0 & \forall j,k\in\{0,1\}\end{cases}.\]

By \cite[Lemma 1, p. 32]{Berger-1971}, we only need to consider the
following cases.

$\bullet$ Case 1: $w_{0|0}=w_{0|1}=0$. In this case, we further have $w_{1|0}=w_{1|1}=1.$
This leads to $R=0$ and $D=2(1-p)\geq D_{\max}$ which is a contradiction
as we only consider $D\in[D_{\min},D_{\max})$.

$\bullet$ Case 2: $w_{1|0}=w_{1|1}=0$. In this case, we have $w_{0|0}=w_{0|1}=1.$
This leads to $R=0$ and $D=1\geq D_{\max}$ which is also a contradiction.

$\bullet$ Case 3: $w_{k|j}>0\,\,\forall j,k\in\{0,1\}$. In this case, we
know $\lambda_{jk}=0$ and then, from $\frac{\partial J}{\partial w_{k|j}}=0$,
we obtain
\begin{align*}
&\bar{p}(\log\frac{w_{k|0}}{q_{k}}+\delta_{0k}\gamma)+\gamma_{0} = 0\quad \forall k\in\{0,1\},\\
& p(\log\frac{w_{k|1}}{q_{k}}+\delta_{1k}\gamma)+\gamma_{1} = 0\quad \forall k\in\{0,1\}.\end{align*}
Equivalently, we have\begin{align*}
&w_{k|0} = q_{k}2^{-\delta_{0k}\gamma}2^{\frac{-\gamma_{0}}{\bar{p}}}\quad\forall k\in\{0,1\},\\
&w_{k|1} = q_{k}2^{-\delta_{1k}\gamma}2^{\frac{-\gamma_{1}}{p}}\quad\forall k\in\{0,1\}.\end{align*}

Letting $\alpha\triangleq2^{-\mu}$ and noticing that $w_{0|j}+w_{1|j}=1\,\,\forall j\in\{0,1\}$,
we get
\begin{align*}
 w_{0|0}=\frac{q_{0}}{q_{0}+q_{1}\alpha},\quad w_{0|1}=\frac{q_{0}\alpha}{q_{0}\alpha+q_{1}},\\
 w_{1|0}=\frac{q_{1}\alpha}{q_{0}+q_{1}\alpha}, \quad w_{1|1}=\frac{q_{1}}{q_{0}\alpha+q_{1}}.
\end{align*}


Putting this into the constraints\[
\begin{cases}
\bar{p}w_{0|0}+pw_{0|1}+2\bar{p}w_{1|0}=D\\
q_{0}=\bar{p}w_{0|0}+pw_{0|1} \\
q_{1}=\bar{p}w_{1|0}+pw_{1|1}\end{cases}\]
we have a set of 3 equations involving 3 variables $\alpha,q_{0},q_{1}$.
Solving this gives us\begin{align*}
&\alpha = \frac{D+p-1}{2-(D+p)},\\
&q_{0} = \frac{2(1-p)-D}{3-2(D+p)},\\
&q_{1} = \frac{1-D}{3-2(D+p)}.\end{align*}

Therefore, we can obtain the optimizing $w_{k|j}$ and have 
\begin{align*}
&R=H(p)-H(\frac{1}{1+\alpha})\\
&\phantom{R}= H(p)-H(D+p-1).
\end{align*}

Hence, in all cases $R=\left[H(p)-H(D+p-1)\right]^{+}$ and we conclude
the proof.
\end{proof}

\section{Proof of Theorem \ref{thm:(BMA-RD)}\label{sec:ProofThmBMARD}}
\begin{proof}
The objective here is to compute the RD function for a discrete source
sequence $x^{N}$ of i.n.d. source components $x_{i}$. First, with
the notations $p_{i,j}\triangleq\Pr(X_{i}=j)$ and $q_{i,j}\triangleq\Pr(\hat{X}_{i}=j)$
for $j\in\{0,1)$ and $i\in\{1,2,\ldots N\},$ Lemma \ref{lem:OneVar}
gives us the rate-distortion components\[
R_{i}(D_{i})=\left[H(p_{i})-H(D_{i}+p_{i,1}-1)\right]^{+}\]
along with the test-channel input-probability distributions
\begin{align*}
q_{i,0}=\frac{2(1-p_{i,1})-D_{i}}{3-2(p_{i,1}+D_{i})} \quad \mbox{and} \quad q_{i,1}=\frac{1-D_{i}}{3-2(p_{i,1}+D_{i})}
\end{align*}
for each index $i$ of the codeword. The overall rate-distortion function
is given by 
\begin{align*}
&R(D)=\min_{\sum_{i=1}^{N}D_{i}=D}R_{i}(D_{i})\\
&\phantom{R(D)}=\min_{\sum_{i=1}^{N}D_{i}=D}\sum_{i=1}^{N}\left[H(p_{i})-H(D_{i}+p_{i,1}-1)\right]^{+}
\end{align*}
which is a convex optimization problem.

Using Lagrange multipliers, we form the functional
\begin{align*}
&J(D)=\sum_{i=1}^{N}\left(H(p_{i,1})-H(D_{i}+p_{i,1}-1)\right)+\gamma\left(\sum_{i=1}^{N}D_{i}-D\right)
\end{align*}
and compute the derivatives \[\frac{\partial J}{\partial D_{i}}=\log(\frac{D_{i}+p_{i,1}-1}{2-D_{i}-p_{i,1}})+\gamma.\]

The Kuhn-Tucker condition (see the restated version in \cite{Gallager-1968},
page 86) then tells us that there is $\gamma$ such that \[
\frac{\partial J}{\partial D_{i}}\begin{cases}
=0 & \mbox{if\,}R_{i}(D_{i})>0\\
\leq0 & \mbox{if\,}R_{i}(D_{i})=0\end{cases}\]
which is equivalent to\[
\frac{D_{i}+p_{i,1}-1}{2-D_{i}-p_{i,1}}\begin{cases}
=2^{-\gamma} & \mbox{if\,}H(p_{i,1})-H(D_{i}+p_{i,1}-1)>0\\
\leq2^{-\gamma} & \mbox{if\,}H(p_{i,1})-H(D_{i}+p_{i,1}-1)\leq0\end{cases}.\]

With the notations $\tilde{D}_{i}\triangleq D_{i}+p_{i,1}-1$ and
$\lambda\triangleq\frac{2^{-\gamma}}{1+2^{-\gamma}}$ , it is equivalent
to
\[
\tilde{D}_{i}\begin{cases}
=\lambda & \mbox{if\,}\tilde{D}_{i}<\min\{p_{i,1},1-p_{i,1}\}\\
\leq\lambda & \mbox{otherwise}\end{cases}.\]

Finally, it becomes \[
\tilde{D}_{i}=\begin{cases}
\lambda & \mbox{if}\,\lambda<\min\{p_{i,1},1-p_{i,1}\}\\
\min\{p_{i,1},1-p_{i,1}\} & \mbox{otherwise}\end{cases}\]
where 
\begin{align*}
&\sum_{i=1}^{N}\tilde{D}_{i}=\sum_{i=1}^{N}(D_{i}+p_{i,1}-1)\\
&\phantom{\sum_{i=1}^{N}\tilde{D}_{i}}=D+\sum_{i=1}^{N}p_{i,1}-N 
\end{align*}
and we conclude the proof.
\end{proof}

\section{Analysis of RDE Computation\label{sec:RDEanalysis}}

Consider a binary single source $X$ with $\Pr(X=1)=p$ and $\Pr(X=0)=1-p\triangleq\bar{p}$.
According to \cite{Blahut-it74}, for any admissible $(R,D)$ pair
we can find two parameters $s\geq0$ and $t\leq0$ so that $F(R,D)$
can be parametrically evaluated as
\begin{align*}
&F(R,D) = sR-stD+\max_{q_{1}}\left(-\log f(q_{1})\right)\\
&\phantom{F(R,D)}= sR-stD-\log\min_{q_{1}}f(q_{1})\end{align*}
where
\[f(q_{1})=\bar{p}\left(\sum_{k}q_{k}2^{t\delta_{0k}}\right)^{-s}+p\left(\sum_{k}q_{k}2^{t\delta_{1k}}\right)^{-s}\]
and $R,D$ are given in terms of optimizing $\underline{q}^{\star}$.

For the distortion measure in (\ref{eq:dstfnBMA}) and with $q_{0}=1-q_{1}$,
we have \[
f(q_{1})=\bar{p}\left((1-q_{1})2^{t}+q_{1}2^{2t}\right)^{-s}+p\left((1-q_{1})2^{t}+q_{1}\right)^{-s}\]
which is a convex function in $q_{1}$. Taking the derivative $\frac{\partial f}{\partial q_{1}}=0$ gives us
\[q_{1}^{\star}=\frac{1+2^{t}}{1-2^{t}}\left(\frac{1}{1+2^{t}}-\frac{\bar{p}^{\frac{1}{s+1}}}{2^{\frac{st}{s+1}}p^{\frac{1}{s+1}}+\bar{p}^{\frac{1}{s+1}}}\right)\triangleq\beta.\]

In order to minimize $f(q_{1})$ over $q_{1}\in[0,1]$, we consider
three following cases where the optimal $q_{1}^{\star}$ is either
on the boundary or at a point with zero gradient.

$\bullet$ Case 1: $0\leq p\leq\frac{2^{t}}{1+2^{t}}$ then $\beta\leq0$.
Since $f$ convex, it is non-decreasing in the interval $[\beta,\infty)$
and therefore in the interval $[0,1]$. Thus, the optimal $q_{1}^{\star}=0$
and we can also compute \[
D=1;\quad R=0;\quad  F=0=D_{KL}(p||p).\]

$\bullet$ Case 2: $1\geq p\geq\frac{1}{1+2^{t(2s+1)}}$ then $\beta\geq1$.
Since $f$ convex, it is non-increasing in the interval $(-\infty,\beta]$
and therefore in the interval $[0,1]$. Thus, the optimal $q_{1}^{\star}=1$
and we get
\[
D=\frac{2\bar{p}}{p2^{2ts}+\bar{p}}; \quad R=0; \quad F=D_{KL}(u||p)\]
where in this case $u=1-\frac{D}{2}$. We can further see that $D\in[2(1-p),1]$
and $u\in[1-D,p]$. 

$\bullet$ Case 3: $\frac{2^{t}}{1+2^{t}}<p<\frac{1}{1+2^{t(2s+1)}}$
then $\beta\in(0,1)$. In this case, the optimal $q_{1}^{\star}=\beta$.
We can find $w_{k|j}^{\star}=\frac{q_{k}^{\star}2^{t\delta_{jk}}}{\sum_{k}q_{k}^{\star}2^{t\delta_{jk}}}$
according to \cite{Blahut-it74} and then obtain
\begin{align*}
&D = \frac{2^{t}}{1+2^{t}}+1-u,\\
&R = H(u)-H(u+D-1),\\
&F = D_{KL}(u||p)\end{align*}
where \[u=\frac{2^{\frac{st}{s+1}}p^{\frac{1}{s+1}}}{2^{\frac{st}{s+1}}p^{\frac{1}{s+1}}+\bar{p}^{\frac{1}{s+1}}}.\]
With this notation of $u$, we can express 
\[q_{1}^{\star}=\frac{1-D}{3-2(u+D)}\quad
\text{and}\quad q_{0}^{\star}=\frac{2(1-u)-D}{3-2(u+D)}.\]
 We can see that $D\in(1-p,1)$.
It can also be verified that, in this case, by varying $s$ and $t,$
$u$ spans $(1-D,1-\frac{D}{2})$ and $R$ spans $(0,H(1-D))$.

\section*{Acknowledgement}

The authors would like to acknowledge the support of Seagate through
the NSF GOALI Program and thank Fatih Erden and Xinmiao Zhang for
valuable discussions on this topic.  The authors are also grateful to
the associate editor and anonymous reviewers for comments that
improved the quality of the paper.

\end{document}